\documentclass[journal]{IEEEtran}
\usepackage{graphicx}
\usepackage{amsmath}
\usepackage[noend]{algpseudocode}
\usepackage{algorithmicx,algorithm}
\usepackage{amsthm}
\usepackage{amsfonts}
\usepackage{subfigure} 
\usepackage{color}
\usepackage{cite}
\usepackage{setspace}
\newtheorem{Lemma}{Lemma}
\newtheorem{Theorem}{Theorem}
\usepackage{amssymb}
\usepackage{color}
\graphicspath{{images/}}
\usepackage{booktabs}
\usepackage{multirow} 
\usepackage{makecell}
 \usepackage[numbers,sort&compress]{natbib}
 \newtheorem{remark}{Remark}
\usepackage{stfloats}
\usepackage{amsmath}
\usepackage{amssymb}
\usepackage{enumitem}
\usepackage{threeparttable}
\usepackage{comment}

\usepackage{bm}

\begin{document}	
\title{ Stochastic Geometry Based Modelling and Analysis of Uplink Cooperative Satellite-Aerial-Terrestrial Networks for Nomadic Communications with Weak Satellite Coverage}
	
\author{Wen-Yu Dong,~\IEEEmembership{Student Member,~IEEE}, Shaoshi Yang*,~\IEEEmembership{Senior Member,~IEEE},\\ Ping Zhang,~\IEEEmembership{Fellow,~IEEE}, Sheng Chen,~\IEEEmembership{Life Fellow,~IEEE}	%
\thanks{This work was supported in part by the Beijing Municipal Natural Science Foundation (No. L242013 and No. Z220004), in part by the Open Project Program of the Key Laboratory of Mathematics and Information Networks, Ministry of Education, China (No. KF202301). \textit{(* Corresponding author)}}
\thanks{W.-Y. Dong and S. Yang are with the School of Information and Communication Engineering, Beijing University of Posts and Telecommunications, with the Key Laboratory of Universal Wireless Communications, Ministry of Education, and also with the Key Laboratory of Mathematics and Information Networks, Ministry of Education, Beijing 100876, China (E-mails: wenyu.dong@bupt.edu.cn, shaoshi.yang@bupt.edu.cn).}
\thanks{P. Zhang is with the School of Information and Communication Engineering, Beijing University of Posts and Telecommunications, and also with the State Key Laboratory of Networking and Switching Technology, Beijing 100876, China (E-mail: pzhang@bupt.edu.cn).} %
\thanks{S. Chen is with the School of Electronics and Computer Science, University of Southampton, Southampton SO17 1BJ, U.K. (E-mail: sqc@ecs.soton.ac.uk).} %
\vspace*{-5mm}
}
	\markboth{Accepted to publish on IEEE Journal on Selected Areas in Communications, Aug. 2024}%
	{Shell \MakeLowercase{\textit{et al.}}: Stochastic Geometry Based Modeling and Analysis of Cooperative Satellite-Aerial-Terrestrial Systems for Nomadic Communications with Weak Signal Coverage}
	
	
\maketitle 

\begin{abstract}
Cooperative satellite-aerial-terrestrial networks (CSATNs), where unmanned aerial vehicles (UAVs) are utilized as nomadic aerial relays (A), are highly valuable for many important applications, such as post-disaster urban reconstruction. In this scenario, direct communication between terrestrial terminals (T) and satellites (S) is often unavailable due to poor propagation conditions for satellite signals, and users tend to congregate in regions of finite size. There is a current dearth in the open literature regarding the uplink performance analysis of CSATN operating under the above constraints, and the few contributions on the uplink model terrestrial terminals by a Poisson point process (PPP) relying on the unrealistic assumption of an infinite area.  This paper aims to fill the above research gap. First, we propose a stochastic geometry based innovative model to characterize the impact of the finite-size distribution region of terrestrial terminals in the CSATN by jointly using a binomial point process (BPP) and a type-II Mat{\'e}rn hard-core point process (MHCPP). Then, we analyze the relationship between the spatial distribution of the coverage areas of aerial nodes and the finite-size distribution region of terrestrial terminals, thereby deriving the distance distribution of the T-A links. Furthermore, we consider the stochastic nature of the spatial distributions of terrestrial terminals and UAVs, and conduct a thorough analysis of the coverage probability and average ergodic rate of the T-A links under Nakagami fading and the A-S links under shadowed-Rician fading. Finally, the accuracy of our theoretical derivations are confirmed by Monte Carlo simulations. Our research offers fundamental insights into the system-level performance optimization for the realistic CSATNs involving nomadic aerial relays and terrestrial terminals confined in a finite-size region.
\end{abstract}
	

\section{Introduction}\label{S1}

\IEEEPARstart{T}{he} sixth generation (6G) mobile communications are aiming at providing ubiquitous connectivity for human society. However, extending the current terrestrial communication infrastructure to cover the vast rural and remote areas is encountering numerous problems, particularly in poor countries \cite{9042251}. In order to realize ubiquitous connectivity, satellite communication has been considered as an alternative solution due to its convenient deployment, significant adaptability, and extensive coverage \cite{9314201, 8353853, 9693912}. Nevertheless, despite the advantages of employing satellites as a supplementary means of communication alongside ground-based systems, practical implementation is impeded by the considerable challenges associated with ultra-long distance transmission from satellites to the earth's surface as well as the obstructive effects of buildings and mountains that result in deep shadowing, thereby hindering direct space-to-ground communication at the data rates expected by ground users.

In order to alleviate the aforementioned challenges, satellite networks have been integrated with terrestrial networks as a viable strategy, commonly known as cooperative satellite-terrestrial networks (CSTNs) \cite{1522108,9179999,9610113}. In CSTNs, terrestrial stations traditionally act as relays, facilitating communication between satellites and terrestrial terminals. This not only extends the coverage of satellite signals but also reduces their pathloss, thereby enhancing the quality of reception at terrestrial terminals. Nevertheless, terrestrial stations also face challenges in fulfilling space-to-ground communication requirements due to various reasons, e.g., terrain conditions. In this context, the utilization of unmanned aerial vehicles (UAVs) as aerial relays is gaining popularity due to its cost-effectiveness and adaptability. For example, in disaster-stricken areas, UAVs can quickly serve as substitutes for impaired terrestrial stations. Moreover, UAVs possess the capability of intelligently adapting their locations to effectively respond to unforeseen traffic requirements in the terrestrial network \cite{9003405}, thus achieving nomadic communication.  Hence, the integration of UAVs into CSTNs, leading to the establishment of cooperative satellite-aerial-terrestrial networks (CSATNs), holds significant importance.

\subsection{Related Works}\label{S1.1}

Thus far, numerous studies have been conducted on the design and analysis of satellite networks. These studies encompass a wide range of topics including system architecture, resource management, security, and performance evaluation, all aimed at enhancing the efficiency, reliability, and/or overall performance of satellite networks. 
	Among the existing studies, some considered the uplink performance analysis of a simpler satellite-terrestrial network,  by deriving the outage probability (OP) \cite{9347980,9509510,9676997,9838778}, while the others considered the downlink ergodic sum rate (ESR) \cite{ huaicong2022ergodic, 9808306} and OP \cite{9130899, song2022cooperative, 9789274,  9582189} for the more complex CSATNs.



	From the information theoretic perspective, ESR and OP are two main metrics for measuring network performance. Specifically, ESR assesses the network’s total throughput, which corresponds to the network-level transmission efficiency, while OP assesses the network’s capability to operate properly under noise and/or interference, which corresponds to the network-level transmission reliability. Therefore, it is essential to investigate the ESR and OP of CSATN systems. It should be noted that there exist other performance metrics for a network, such as end-to-end latency, network jitter, network lifetime, and so forth. However, these metrics are typically studied on the network layer, rather than on the fundamental physical layer of our interest. In \cite{huaicong2022ergodic}, a switch-based hybrid FSO/millimeter wave scheme with a robust beamforming (BF) algorithm was proposed for the uplink non-orthogonal multiple access scenario, while a hybrid multiple access scheme was suggested in \cite{9808306} to offer dependable connectivity for heterogeneous users in a CSATN. 	Yastrebova et al.  \cite{9347980} investigated the impact of terrestrial interference on the uplink of an LEO satellite constellation in high frequency bands of International Mobile Telecommunications (IMT), and the uplink coverage probability in hybrid satellite-terrestrial networks was analyzed in \cite{9509510}. Manzoor et al. \cite{9676997} modeled the data frame repetition behaviour based on the probability of a single transmission taking place on line-of-sight (LoS) links, and analyzed the coverage performance based on the frame success rate. Furthermore, Chan et al.  \cite{9838778} analyzed the uplink coverage probability and uplink throughput performance in massive IoT-over-satellite networks. 
	 The downlink OP of a CSATN was investigated in \cite{9130899}, where an optimization problem in terms of the transmit power and the transmission time over the satellite-aerial and aerial-terrestrial links is formulated and solved, to obtain the optimal end-to-end energy efficiency for the CSATN system considered. Song et al. \cite{song2022cooperative} derived a closed-form expression for the OP of several types of communication links, including the uplink from the aerial source to the satellite or aerial relay, the downlink from the satellite or aerial relay to the terrestrial destination, and the links between aerial relays.  Zhang et al. \cite{9789274} conducted a study on the OP of the downlink in a CSATN system that utilizes the decode-and-forward (DF) relay mechanism. Vasudha et al. \cite{9582189} considered the relative velocity and the random position of receivers, and analyzed the aerial-terrestrial downlink performance by evaluating the instantaneous signal-to-noise ratio in the presence of erroneous estimation of both channel gain and noise.

 We emphasize that among the above contributions, \cite{9347980,9509510,9676997,9838778} focused on analyzing the uplink performance; however, their system models only involved ground devices and satellites, without including aerial devices. Additionally, most existing studies on the performance of CSATNs focus on the downlink \cite{9130899, song2022cooperative, 9789274,  9582189}, while the analysis of the uplink of CSATNs is lacking. In this paper we aim to fill this gap by providing a comprehensive analysis of the uplink performance of a CSATN that utilizes UAVs as relays.

In various real world scenarios, such as post-disaster urban reconstruction, it is commonly seen that users tend to congregate in finite-size areas. Understanding this fundamental property is critical for more effective network restoration or achieving more efficient network operation. However, evidently a significant gap exists in the current research landscape regarding the uplink performance analysis of CSATN systems, where the terrestrial terminals are located in finite-size areas.  Existing studies often model terrestrial terminals by a Poisson point process (PPP) that relies on the assumption of an infinite distribution region \cite{9347980,9509510,9676997,9838778}, which is not suitable for finite-size areas due to its boundary effects. As both the binomial point process (BPP) \cite{7882710} and type-II Mat{\'e}rn hard-core point process (MHCPP) \cite{6574907} are mathematically defined within finite regions, they are well suited for areas with spatial boundaries. This makes them particularly effective for explaining the network configuration attributes of finite-size areas.
 Although many studies using stochastic geometry explore finite distribution regions, they predominantly focus on terrestrial device-to-device networks \cite{6932503,8016632,Zhao_etal2019}, while neglecting the unique challenges posed by CSATNs. The main motivation of our work is to address this void in the research concerning CSATNs. 


In addition to the aforementioned scenario modeling, the channel model also constitutes a pivotal aspect of CSATNs. The shadowed-Rician (SR) fading, distinguished from Rayleigh and Nakagami fading, has proven to be a more fitting choice for statistically characterizing satellite channels, as demonstrated in various frequency bands, e.g., S-, L-, Ku-, and Ka-band \cite{1623307}. This versatility positions SR fading as a well-suited option for modeling satellite communication channels. In the existing works that focused on satellite communication using the SR fading model \cite{bhatnagar2013closed, 7869087, 8894851, 9678973, 8068989, 2021Stochastic}, however, either the presence of interference is neglected or the channel model for interference links is studied using the Nakagami fading model. Therefore, it is evident that a comprehensive investigation into interference using the SR model is lacking. Addressing this issue is another motivation of our work.
              
\subsection{Contributions and Organization of the Paper}\label{S1.2}
\begin{table*}[!t]
\caption{Comparison of existing state-of-the-arts in related works with our proposed work} 
\label{Comparison} 
\vspace*{-5.5mm}
\begin{center}
\resizebox{0.93\linewidth}{!}{
\begin{threeparttable}
\begin{tabular}{c|c|c|c|c|c|c} 
	\bottomrule
	\textbf{Reference} & \textbf{Link types} & \textbf{\makecell{Channel fading model \\ of target signal}} & \textbf{\makecell{Channel fading model \\ of interference signal}} & \textbf{finite-size area} & \textbf{Beamforming} & \textbf{Point process} \\ \bottomrule
		Yastrebova \emph{et al.} \cite{9347980}& T-S & T-S: Rician & T-S: Rician  & - & - & PPP \\ \hline\rule{0pt}{8pt}	
	Homssi \emph{et al.} \cite{9509510}& T-BS-S & \makecell{T-BS: Rayleigh \\ BS-S: the empirical model} & \makecell{T-BS: Rayleigh \\ BS-S: the empirical model} & - & - & PPP \\ \hline\rule{0pt}{8pt}
	 Manzoor \emph{et al.}	\cite{9676997} & T-S & T-S: the empirical model & T-S: the empirical model & - & - & PPP \\ \hline\rule{0pt}{8pt}
	Chan \emph{et al.} \cite{9838778}& T-S & T-S: the empirical model & T-S: the empirical model & - & - & PPP \\ \hline\rule{0pt}{8pt}
	Pan \emph{et al.} \cite{9130899} & S-A-T & \makecell{S-A: SR \\ A-T: Rician} & - & -	& - & Single-Point \\ \hline\rule{0pt}{8pt}
	Song \emph{et al.} \cite{song2022cooperative} & A-S-GW & \makecell{A-S: Nakagami \\ S-GW: SR} & Nakagami & - & - &	MHCPP/PPP \\ \hline\rule{0pt}{8pt}
	Zhang \emph{et al.}	\cite{ 9789274} & S-A-T & \makecell{S-A: SR \\ A-T: Rician} & - & - & - & Single-Point \\ \hline\rule{0pt}{8pt}
	Vasudha \emph{et al.} \cite{9582189}& S-A-T & \makecell{S-A:~- \\ A-T: Rayleigh} & - & - & - & PPP \\ \hline\rule{0pt}{8pt}
	Kolawole \emph{et al.} \cite{8068989}& \makecell{S-T\\BS-T} &  \makecell{S-T: SR (Approximated\\ as a Gamma function) } & BS-T: Nakagami & - & \checkmark & PPP \\ \hline\rule{0pt}{8pt}
	Talgat \emph{et al.} \cite{2021Stochastic} & S-GW-T & \makecell{S-GW: SR \\ GW-T: Rayleigh}& - & - & - & BPP/PPP\\ \hline\rule{0pt}{8pt}
	Jung \emph{et al.} \cite{9678973} & S-T & SR  & - & - & - & BPP \\ \hline\rule{0pt}{8pt}
	Our work & T-A-S &  \makecell{T-A: Nakagami \\ A-S: SR}& \makecell{T-A: Nakagami \\ A-S: SR} & \checkmark & \checkmark & BPP/MHCPP \\ \bottomrule
\end{tabular}
\begin{tablenotes}
  \footnotesize
	\item S: satellite, A: aerial node, T: terrestrial terminal, GW: gateway, BS: base station. SR: shadowed-Rician fading. \checkmark: considered, - : not considered.
	\item  MHCPP: Mat{\'e}rn hard-core point process, PPP: Poisson point process, BPP: binomial point process.
\end{tablenotes}
\end{threeparttable}
}
\end{center}
\vspace*{-7mm}
\end{table*}
	 
Inspired by the insights gained from prior discoveries, this study directs its focus towards the uplink performance of CSATNs, whose terrestrial terminals are constrained in finite-size regions. Specifically, we introduce an innovative CSATN system tailored to a finite-size region and derive an expression to characterize the communication coverage of users within this finite-size area, where terrestrial terminals establish connections with satellites through the assistance of multiple aerial relays. In our exploration, we partition the uplink system into two main segments: the terrestrial terminals to aerial relays (T-A) link and the aerial relays to the satellite (A-S) link. Subsequently, we conduct an  accurate analysis of each segment's performance by using Nakagami and SR fading models, respectively. This approach allows us to gain profound insights into the network's behavior, shedding light on the unique challenges and opportunities presented by the interplay of terrestrial, aerial, and satellite components in a finite-size geographical region. Our novel contributions are summarized as follows.
	 
\begin{itemize}[label=\textbullet,font=\Large]
\item We introduce a stochastic geometry based innovative model to characterize the impact of the finite-size distribution region of terrestrial terminals in the CSATN by jointly using a BPP  and a type-II MHCPP. There is a current dearth in the open literature regarding the uplink performance analysis of CSATN. The limited research of using stochastic geometry in satellite communications mainly focuses on the downlink of CSTNs and uses the classical PPP to model the distribution of terrestrial nodes. However, PPP relies on the unrealistic assumption that terrestrial nodes are distributed in an infinite area. In this work, the CSATN considered comprises a satellite, multiple aerial nodes, and a set of terrestrial terminals that are located in a finite-size area. In our proposed model, terrestrial terminals are governed by a BPP, while aerial nodes adhere to a type-II MHCPP. 

\item We analyze the sophisticated relationship between the spatial distribution of the coverage areas of aerial nodes and the finite-size distribution region of terrestrial terminals, thereby deriving the T-A links' distance distribution, which must be obtained for further analyzing the coverage probability and the average ergodic rate. Our system model incorporates a representation of the coverage area of each aerial node, in the form of a circle that results from a 2-D projection of a cone. Although there exist studies that use stochastic geometry to analyze network performance subject to the constraint of finite node distribution region, they predominantly focus on terrestrial networks, where the unique challenges posed by CSATNs are of course not involved. Furthermore, we delineate the operational range of the aerial nodes by assessing their communication links with the terrestrial terminals.  

\item We consider the stochastic nature of the spatial distributions of terrestrial terminals and UAVs, and conduct a thorough analysis of  both the coverage probability and the average ergodic rate of the T-A links under Nakagami fading and the A-S links under SR fading. Since there is still a lack of interference analysis for the SR model at present, we propose a novel method for analyzing the coverage probability and the average ergodic rate of A-S links while assuming the interference imposed on the target aerial node by other aerial nodes experiences the SR fading. Specifically, we delineate the statistical properties of the received signal-to-interference-plus-noise ratio (SINR) for the T-A and A-S links based on the distance distribution function obtained. Moreover, we analyze the communication quality of the A-S links in terms of coverage probability when using the directional BF, and the benefits of this strategy are validated by subsequent simulations. 

\item We garner a substantial volume of results to assess the performance of the CSATN considered via extensive Monte Carlo simulations, which demonstrate the correctness of our theoretical analysis.  Specifically, we conduct comprehensive numerical simulations and discussions for both the T-A and A-S links. Additionally, we conduct a comparative study of the impact of key system parameters, including the coverage area of aerial nodes, the flying altitude of aerial nodes, as well as the terrestrial and aerial nodes' densities, transmission distance, and antenna gain, on the achievable system performance.
\end{itemize}		 
	 
Table~\ref{Comparison} summarizes the distinctions between our work and the state-of-the-art studies, which highlights the novelty of our contributions.

The remainder of this paper is organized as follows. Section~\ref{S2} presents the network deployment, propagation models and association policy for the CSATN system. The distribution of distances within a finite-size region is provided in Section~\ref{S3}, where we discuss the positioning relationship between the finite-size area and the coverage range of aerial nodes. Section~\ref{S4} presents our main performance analysis results, including the derivation of the analytical coverage probability for the T-A link and A-S link in the uplink. Section~\ref{S5} provides the derivations of the analytical average ergodic rate for the T-A link and the A-S link, respectively. In Section~\ref{S6}, we provide numerical results to verify our theoretical derivations and to study the effect of key system parameters, such as network scale, array gain and transmission power, on the network performance. Our conclusions are drawn in Section~\ref{S7}.

\textbf{Notation:} $\mathbb{P}(\cdot)$ denotes the probability measure and $\mathbb{E}[\cdot]$ denotes the average measure. The Laplace transform of random variable $X$ is defined as $\mathcal{L}_X\left(s\right)
=\mathbb{E}\,\left[\exp (-s X)\right]$. The cumulative distribution function (CDF) and probability density function (PDF) of $X$ are denoted by $F_X(x)$ and $f_X(x)$, respectively, while the conditional PDF of $X$ conditioned on $Y$ is denoted as $f_{X|Y}(x|y)$. $\Gamma(\cdot)$ is the Gamma function, and $(\cdot)_{n}$ is the Pochhammer symbol, which is defined as $(x)_{n}=\Gamma(x+n)/\Gamma(x)$. The lower incomplete Gamma function is defined as $\bar{\gamma}(a,x)=\int_0^xt^{a-1}\exp (-t){\rm d}t$. $\mathcal{B}(\bm{o},R_{\textrm{A}})$ denotes the circular plane with radius $R_{\textrm{A}}$ centered at point $\bm{o}$. The 2-norm of $\bm{x}=[x_1, x_2, \cdots , x_n]^{\rm T}$ is defined as $\left\|\bm{x}\right\|_2=\sqrt{x_1^2+\cdots+x_n^2}$. $\binom n k$ denotes the binomial coefficient. ${_1F_1}\left(\cdot;\cdot;\cdot\right)$ is the confluent hypergeometric function of the first kind.

\section{System Model}\label{S2}
 
The system model of our uplink CSATN is depicted in Fig.~\ref{fig:1}, which consists of an LEO satellite ($S$), a group of aerial nodes, e.g., UAVs ($A_n$, $n\in\{1,\cdots,N_{\textrm{A}}\}$ and $N_{\textrm{A}}\geq 1$) and a number of terrestrial terminals, i.e., users, ($T_l$, $l\in\{1,\cdots,N_0\}$ and $N_0\geq 1$). The aerial-satellite links operate in frequency division duplex (FDD) mode to improve the communication efficiency, mitigate sophisticated interference, and simplify performance assessment\footnote{In satellite communications, significant round-trip signal delay due to long signal propagation distance causes a large idle period between the pair of uplink and downlink time slots when using time division duplex (TDD), making TDD inefficient for most satellite communication systems. Currently almost all the existing satellite communication systems adopt FDD rather than TDD. In addition, time-slot overlap due to time synchronization errors may lead to sophisticated interference. By contrast, FDD can effectively mitigate these issues, though the individual bandwidths on the uplink and downlink are less abundant.}.  In terms of the terrestrial-aerial links, since the transmission distance is much smaller than that of the aerial-satellite links, in principle both FDD and TDD can be used, depending on specific application requirements. Specifically, the signal from $T_l$ is transmitted to $S$ through $A_n$ using two links, namely, the T-A link and the A-S link.
 
\begin{figure}[!tp]
\begin{center}
\includegraphics[width=0.9\columnwidth]{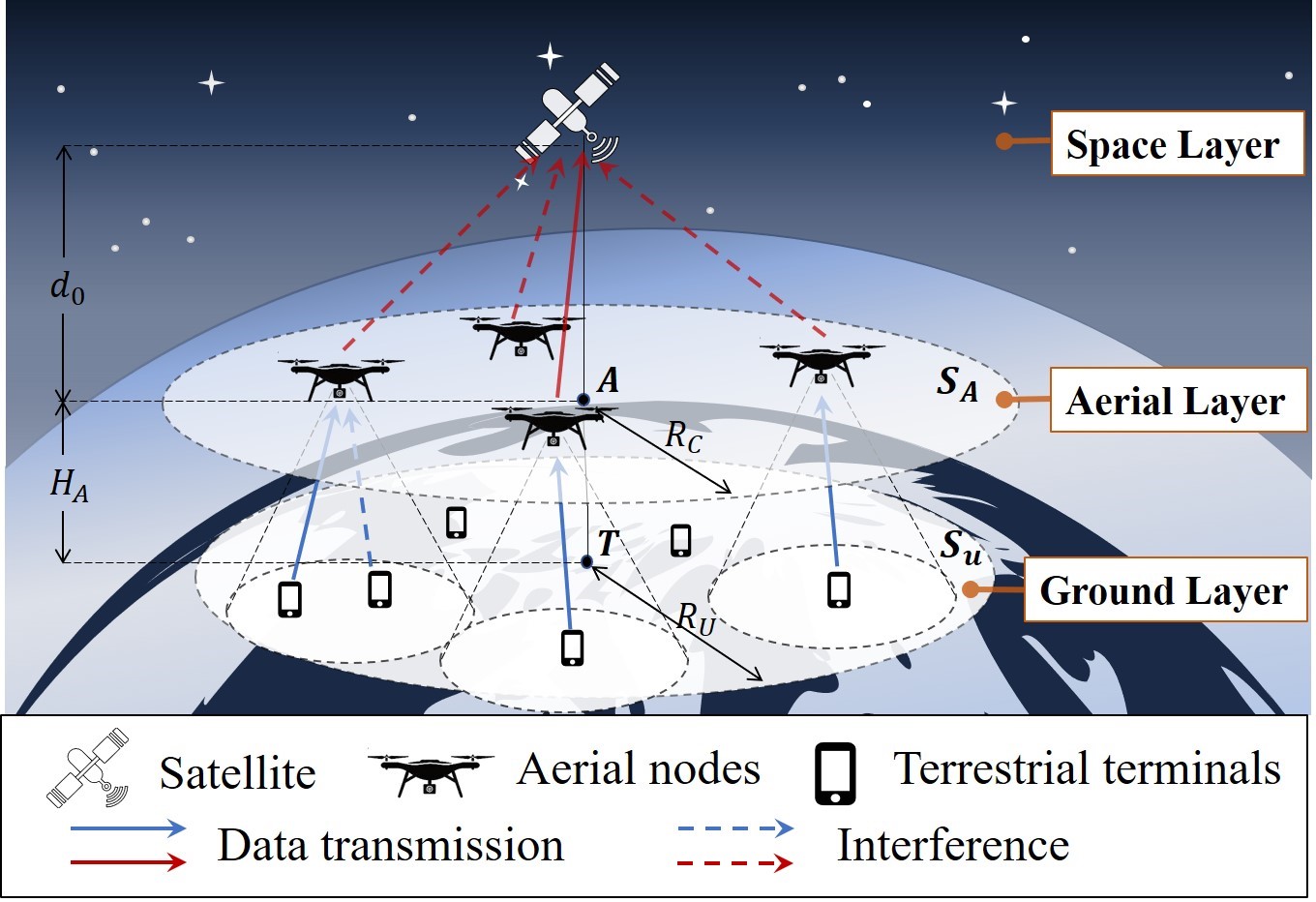}
\end{center}
\vspace*{-4mm}
\caption{Illustration of the CSATN system.}
\label{fig:1}
\vspace*{-5mm}
\end{figure}
The aerial nodes serve as relays that can use any feasible relay protocol, such as amplify-and-forward (AF) and DF, while the satellite either processes the data as a space base station or as a space relay forwarding ground signals to remote ground gateways.  In order to enhance spectrum utilization, the communications on the T-A links, and the communications on the A-S links, share their individual frequency bands, respectively. As a result, the received signals at $A_n$ and $S$ are subject to multi-user interference (MUI) caused by other terrestrial terminals and other aerial nodes, respectively. It is worth noting that the aerial nodes employ a dual-band, dual radio frequency (DRF) transmission and reception mechanism. Hence, the transmission of signals from terrestrial terminals will not cause any interference with the transmission of signals from aerial nodes to the satellite. {\color{black}Furthermore, it is assumed that aerial nodes in this system fly at the same height $H_{\textrm{A}}$ above the earth surface while the satellite is located at a height of $d_0$ above the plane of aerial nodes\footnote{{\color{black}In practical applications, slight errors in the height of aerial nodes are inevitable, making it unrealistic for them to be perfectly aligned at the same altitude. However, in ultra-low altitude environments, the impact of small fluctuations in $H_{\textrm{A}}$ on the theoretic results can be neglected. Therefore, for the sake of simplicity, we assume that the aerial nodes are positioned at the same height.}}.} For the sake of conciseness, we also use A or AN and T or TN to stand for aerial node and terrestrial node, respectively.

\subsection{Topology Deployment}\label{S2.1}

\subsubsection{Deployment of TNs}

We consider a realistic scenario where the majority of users are geographically concentrated within a specific finite region, such as  a post-disaster gathering point. To represent this user concentration, we use a circular bounded area designated as ${\color{black}\mathcal{S_{\textrm{U}}}}$, with radius $R_{\textrm{U}}$. Within this area, users are modeled as BPP denoted as $\Phi_1$, with density $\lambda_{\textrm{T}}$. Additionally, users located inside the coverage area of an AN  are represented as $\Phi_{\textrm{U}}$.
  
\subsubsection{Deployment of ANs}

We exploit type-II MHCPP to mimic the deployment of ANs. The aerial transmitters are randomly dispersed on the circular plane, denoted as ${\color{black}\mathcal{S_{\textrm{A}}}}$, with radius $R_{\textrm{C}}$.  The vertical distance from the ANs to the ground is $H_{\textrm{A}}$. Additionally, the coverage of each AN $A$ on the ground can be modeled as a circular plane of radius $R_{\textrm{A}}$. That is, the coverage range of AN $A$ on the ground is the circular region ${\color{black}\mathcal{S}'_{\textrm{A}}}={\cal B}(\bm{a},R_{\textrm{A}})$, which contains TNs capable of establishing connection with the AN, where $\bm{a}$ is the ground projection position of $A$.
	 
Next we delineate the construction of MHCPP, represented as $\Phi_{\textrm{A}}$, utilizing a two-step approach.
	
Initially, candidate points are uniformly generated using a homogeneous PPP with density $\lambda_1$ within the circular space ${\color{black}\mathcal{S_{\textrm{A}}}}$. With respect to the quantity of candidate points, denoted as $N_c$, the PDF of the Poisson distribution can be represented as $\mathbb{P}(N_{c}\! =\! k)=\frac{\lambda_{1}S_{\textrm{A}}}{k!}\exp(-\lambda_{1}S_{\textrm{A}})$ where $k$ is the number of candidate points and $S_{\textrm{A}}$ denotes the area of space ${\color{black}\mathcal{S_{\textrm{A}}}}$.   
	
Subsequently, each candidate point is assigned with an individual score, randomly selected from a uniform distribution between 0 and 1.  Afterward, we selectively keep the point with the smallest score within a restricted circular space defined by the radius $D_{\min}$ and eliminate the others. To elaborate, a given point $Q$ initially acts as the center of a small circular repulsion space with the repulsion radius $D_{\min}$, i.e., ${\cal B}(\bm{q},D_{\min})$, where $\bm{q}$ is the position of $Q$. The concept of the repulsion space is that within this space, if there are multiple candidate nodes, some nodes need to be removed so that only a single candidate node finally remains within a circular repulsion space with a repulsion radius of $D_{\min}$. Specifically, if there are additional candidate points within ${\cal B}(\bm{q},D_{\mathrm{min}})$, their scores are compared with that of $Q$, and only the point with the smallest score is retained in each  repulsion space. This implies that the distance $D$ between every two points is larger than or equal to $D_{\mathrm{min}}$.

At the end of the second step, we construct the MHCPP, denoted as $\Phi_{\textrm{A}}$, with the density $\lambda_A$ which can be mathematically represented as $\lambda_{\textrm{A}}=\frac{1-\exp\big(-\pi D_{\mathrm{min}}^{2}\lambda_1\big)}{\pi D_{\mathrm{min}}^{2}}$. 

\vspace{-0.2cm}
\subsection{Terrestrial-Aerial Link Propagation Model}\label{S2.2}
	
\subsubsection{Channel Fading} 

The terrestrial-aerial connection commonly relies on LoS transmission [31]. In addition, terrestrial devices are often situated in outdoor environments prone to rich scattering, hence the signals from these devices to aerial nodes often propagate through multiple paths. Due to the versatility of the Nakagami fading model and the fact that it often better matches empirical data, we use it to effectively characterize the signal propagation characteristics of terrestrial-aerial links in complex environments possibly with a mixture of rich scattering and the LoS paths\footnote{The Rician and the Nakagami models behave approximately equivalently near their mean value. This observation has often been used to advocate the Nakagami model as an approximation for situations where a Rician model would be more appropriate.}.
 Specifically, the small-scaling channel fading coefficient $h_{\textrm{TA}}$ of T-A link exhibits Nakagami fading. Consequently, $|h_{\textrm{TA}}|^2$ can be modeled as a random variable that follows a normalized Gamma distribution. The Nakagami fading parameter is denoted by $N_{\textrm{TA}}$, and for the sake of simplicity, it is assumed to be a positive integer \cite{6932503}.

\subsubsection{SINR model} 

The SINR at the AN receiver $A$ of the T-A link from the target transmitter $T_m$ can be expressed as:
\begin{align}\label{eqSINR} 
	\mathrm{SINR}_1 =& \frac{	P_{\textrm{T}} \left|h_{T_m A}\right|^2 (H_{\textrm{A}}^2 + R_m^2)^{-\frac{\alpha_1}{2}}}{I_{\textrm{T}}+\sigma_{\textrm{T}}^2} \nonumber \\
	\approx & \frac{	P_{\textrm{T}} \left|h_{T_m\!A}\right|^2 (H_{\textrm{A}}^2+R_m^2)^{-\frac{\alpha_1}{2}}} {I_{\textrm{T}}},
\end{align}
where the interference power 
\begin{align}\label{eqInterf} 
  I_{\textrm{T}} =&\sum_{T_n\in\Phi_{\textrm{U}} {\backslash} T_m} 	P_{\textrm{T}}\left|h_{T_n }\right|^2 (H_{\textrm{A}}^2+R_n^2)^{-\frac{\alpha_1}{2}},
\end{align}
$\Phi_{\textrm{U}}=\Phi_1 \bigcap \mathcal{B}(\bm{a},R_{\textrm{A}}) $, $T_n$ are the interfering TNs, $P_{\textrm{T}}$ is the transmit power at TNs, $R_{m}$ and $R_n$ are the distances from the target node $T_m$ and the interference node $T_n$ to the terrestrial projection position $\bm{a}$ of the AN $A$, respectively, while $\alpha_1$ is the path-loss exponent of the T-A link, and $\sigma_{\textrm{T}}^2$ is the strength of additive white Gaussian noise (AWGN) of the T-A link. The approximation in (\ref{eqSINR}) is due to the fact that the system is interference limited and $I_{\textrm{T}} \gg \sigma_{\textrm{T}}^2$.

\subsection{Aerial-Satellite Link Propagation Model}\label{S2.3}

\subsubsection{Directional beamforming modeling} 

We assume that the antenna array on the satellite uses receive BF, and the antenna array on each aerial node uses transmit BF. In principle, the aerial nodes can also use receive BF, but this assumption is not mandatory herein.  The physical antenna array is approximated using a sectored antenna model for simplicity. This assumption is justified because we can use specific BF algorithms and corresponding sidelobe suppression techniques to achieve a single mainlobe with a significantly higher gain than the sidelobes \cite{9926973, 8059815}. Furthermore, since the satellite has a large coverage area, the signals transmitted by all the aerial nodes within the satellite's coverage area are assumed to be received by the satellite's mainlobe. Meanwhile, we assume that for the target aerial node, it always uses the mainlobe to transmit signals to the satellite, whilst for each interfering aerial node, whether its mainlobe or sidelobe aligns to the satellite is random. 
 In the sectored antenna model, it is assumed that the array gain remains a constant value $G$ for all the angles within the mainlobe, and it is a constant value $g$ within the sidelobe\footnote{Although this assumption is not perfectly accurate from theoretic perspective, it is indeed a sufficiently good approximation in most practical applications, as evidenced by the technical specifications of commercial sector antenna products for cellular base stations.}.

Let the overall directional BF gain from an AN $A$ to the satellite $S$ be $G_{\textrm{A-S}}$, $G_{\textrm{r}}$ signify the gain of the receiver array mainlobe, and the angular width of the mainlobe of transmitter be $\theta$. Further denote $G_{\textrm{t}}$ and $g_{\textrm{t}}$ as the transmitter array gains of the mainlobe and sidelobe, respectively. The target AN $A_m$ and its serving satellite $S$ can estimate the angles of departure and arrival, and then adjust their antenna steering orientations to exploit the maximum directivity gain \cite{6932503}. We assume that the estimates are perfect. The value of $G_{A_m\!-\!S}$ for the A-S link from a given target AN $A_m$ to the satellite $S$ is given by $D_{A_m\!-\!S}=G_{\textrm{t}} G_{\textrm{r}}$. On the other hand, for the $n$-th interfering link, we assume that the angle of arrival and the angle of departure of the signal are independently and uniformly distributed in the range $(0, \,2\pi]$, which results in a random directivity gain denoted as $D_{A_n\!-\!S}$. Thus, the probability distribution for $D_{A_n\!-\!S}$ can be expressed as: 
\begin{align}\label{eqDBFg} 
	G_{A_n\!-\!S} =& \left\{ \begin{array}{cl}
		G_{\textrm{t}} G_{\textrm{r}}	, & {\rm P}_{{\rm Mb},{\rm Mb}}=\frac{\theta}{2\pi} , \\ 
	g_{\textrm{t}} G_{\textrm{r}}	, & {\rm P}_{{\rm Sb},{\rm Mb}}=1-\frac{\theta}{2\pi} ,
	\end{array} \right.
\end{align}
where ${\rm P}_{tx,{\rm Mb}}$ denotes the probability of the $n$-th A-S link in state $(tx,{\rm Mb})$, with $tx \in\{{\rm Sb},{\rm Mb}\}$, while $\rm Sb$ and $\rm Mb$ denote the sidelobe and mainlobe.

\subsubsection{Channel fading} 

The SR fading model \cite{2003A} is commonly used to analyze the performance of satellite-terrestrial links in different fixed and mobile satellite services operating in various frequency bands \cite{jung2018outage}. For the low-altitude nomadic relay UAVs in the CSATN system considered, the presence of abundant shadowing effect due to complex terrain, combined with the predominant LoS propagation path, closely aligns with the characteristics of SR fading. In our analysis, $2 c$ denotes the mean power of the multipath component excluding the LoS component, and $\Omega$ denotes the mean power of the LoS component, while $q$ is the Nakagami fading parameter. Let the small-scaling channel fading coefficient of the A-S link be $h$. Then the PDF of $|h|^{2}$ is represented as \cite{zhang2019performance}:	
\begin{align}\label{eqPDFscf} 
	f_{|h|^2}(x) =& \kappa\, \exp(-\beta x)\, {_1F_1}\left(q\, ;1\, ;\delta x\right ),
\end{align}
where $\kappa =\frac{(2 c q)^q}{2 c(2 c q+\Omega)^q}$, $\delta=\frac{\Omega}{2 c(2 c q+\Omega)}$, and $\beta=\frac{1}{2 c}$.
     
\subsubsection{SINR model} 

According to the aforementioned model, the SINR at the satellite receiver $S$ originating from the target AN $A_m$ can be expressed as:
\begin{align}\label{eqSINRa-s} 
  \mathrm{SINR}_2
  \approx& \frac{P_m G_{A_m\!-\!S}\left|h_{A_m\!S}\right|^2 d_{m}^{-\alpha_2}}{I_{\textrm{A}}},
\end{align}
where the interference power
\begin{align}\label{eqInta-s} 
  I_{\textrm{A}} =& \sum_{A_n\in\Phi_{\textrm{A}} {\backslash} A_m} P_{n}G_{A_n\!-\!S}\left|h_{A_n\! S}\right|^2 d_{n}^{-\alpha_2},
\end{align}
$A_n$ are interfering ANs, $P_m$ and $P_n$ are the transmit power of target AN $A_m$ and interfering ANs $A_n$, respectively, $d_{m}$ and $d_{n}$ are the distances between $A_m$ and $S$ and between $A_n$ and $S$, respectively, while $\alpha_2$ is the path-loss exponent of the A-S link, and $ \sigma_{\textrm{A}}^2$ is the strength of AWGN of the A-S link. Again owing to the interference limited nature, $I_{\textrm{A}} \gg \sigma_{\textrm{A}}^2$.
        
\subsection{Association Policy}\label{S2.4}
       
There is a strong repulsion among the points in the MHCPP. We set the ground coverage radius of AN $R_{\textrm{A}}$ to half of the repulsion radius $D_{\textrm{min}}$, i.e., $R_{\textrm{A}}\! =\! \frac{D_{\textrm{min}}}{2}$. Hence, TNs can establish communication with at most one AN. We further explore an association policy where the links between TNs and ANs as well as between ANs and the satellite are randomly established. In other words, our target transmitter is randomly selected within the communication range and the T-A link is independent of the A-S link. Therefore, meeting the following criteria indicates the user's ability to successfully establish communication with the satellite: (i) the $\mathrm{SINR}_1$ for the T-A link surpasses a preset threshold denoted as $T_{h_1}$, and (ii) the $\mathrm{SINR}_2$ for the A-S link surpasses a predefined threshold denoted as $T_{h_2}$. By defining  $P_{\textrm{cov}}^{\textrm{T-A}}\! \triangleq\mathbb{P}\!\left( \mathrm{SINR}_1\! \geq\! T_{h_1}\right)$ and $P_{\textrm{cov}}^\textrm{{A-S}}\! \triangleq\! \mathbb{P}\left( \mathrm{SINR}_2\! \geq\! T_{h_2}\right)$, the coverage probability can be expressed as
\begin{align}\label{eqCPdf} 
	P_{\textrm{cov}} =& P_{\textrm{cov}}^{\textrm{T-A}}\, P_{\textrm{cov}}^{\textrm{A-S}} .
\end{align}

\section{Distribution of Distances}\label{S3}

In order to derive the coverage probability expressions, it is necessary to firstly characterize the distance distributions arising from the stochastic geometry of the system under consideration. In particular, we present the PDF for the distribution of distances in \textbf{Lemma~\ref{L1}} and \textbf{Lemma~\ref{L2}}.

As the BPP is employed to model the user distribution, $N_0$ TNs are scattered independently and uniformly across the circular wireless network centered at $\bm{u}$ with radius $R_{\textrm{U}}$. Consequently, the PDF describing the user distribution within the region $\mathcal{B}(\bm{u},R_{\textrm{U}})$ is given by
\begin{align}\label{Lemma1} 
	f\left(\bm{x}_i\right) =& \frac{1}{\pi R_{\textrm{U}}}^2 , ~ \left\|\bm{x}_i-\bm{u}\right\|_2 \leq R_{\textrm{U}} ,
\end{align}
where $\bm{x}_i$ for $i\in\{1,2, \cdots,N_0\}$ are the locations of users.
	
\begin{figure}[tbp]
\begin{center}
\includegraphics[width=0.9\columnwidth]{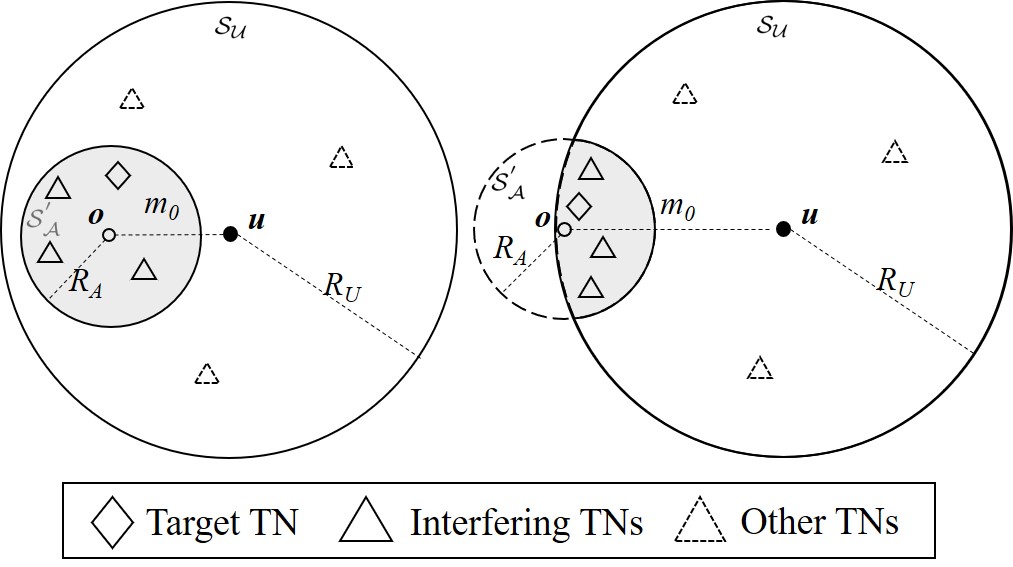}
\end{center}
\vspace*{-4mm}
\caption{The positional relationship between ${\color{black}\mathcal{S}_{\textrm{A}}'}$ and ${\color{black}\mathcal{S_{\textrm{U}}}}$, where $\bm{o}$ is the ground projection position of the AN considered.}
\label{position} 
\vspace*{-5mm}
\end{figure}

From Fig.~\ref{fig:1}, it can be seen that the coverage area of each AN is finite-size, impacting the users eligible to establish communication with an AN. Specifically, these users must be situated inside the intersecting region of ${\color{black}\mathcal{S}_{\textrm{A}}'}$ and ${\color{black}\mathcal{S_{\textrm{U}}}}$. Given the constant vertical separation between ANs and the ground, our focus narrows down to the distribution of the distance $r$ between the user and the AN's projection point on the ground. Let the random variable  $M$ be the distance from the projected location of an AN to the center of the user area. In this context, we can analyze the PDF of the distance $r$ conditioned on $M$ equal to $m_0$. Based on the geometric relationships, it becomes apparent that when $M > R_{\textrm{A}} + R_{\textrm{U}}$, there will be no users within the AN's coverage range. Thus, we assume that the radius $R_{\textrm{C}}$ of all the ANs' location area, ${\color{black}\mathcal{S}_{\textrm{A}}}$, is equivalent to the sum of $R_{\textrm{A}}$ and $R_{\textrm{U}}$, i.e., $R_{\textrm{C}} = R_{\textrm{A}} + R_{\textrm{U}}$.
\begin{figure*}[b]	\setcounter{equation}{9}
	\vspace*{-5mm}
	\hrulefill
	\vspace*{-1.5mm}
	\begin{align}\label{Lemma1.1} 
		f_{R^{3)}}(r|m_0) =& \frac{r\sqrt{4m_0^2R_{\textrm{U}}^2-\left(m_0+R_{\textrm{U}}^2-r^2\right)^2}}{m_0^2} + 2r\left(\varphi_{3}-\frac{1}{2}\sin\big(2\varphi_3\big)\right) \nonumber \\ &
		+\frac{\big(m_{0}^{2}+R_{U}^{2}-r^{2}\big)\sqrt{2R_{\textrm{U}}^2(m_0+r^2)-R_{\textrm{U}}^4-(m_0^2-r^2)^2}}{2m_0^{2}r} .
	\end{align}
	\vspace*{-2mm}
\end{figure*}
	
The PDF of the distance distribution from the users to the ground projection point of an AN can be categorized into two scenarios based on the relative spatial relationship between the coverage range of the AN, ${\color{black}\mathcal{S}_{\textrm{A}}'}$, and the user area, ${\color{black}\mathcal{S_{\textrm{U}}}}$. These two scenarios are illustrated in Fig.~\ref{position}. The first scenario arises when the complete area of ${\color{black}\mathcal{S}_{\textrm{A}}'}$ is entirely within ${\color{black}\mathcal{S_{\textrm{U}}}}$, while the second scenario occurs when these two areas have a partial intersection. Further details can be found in Appendix~\ref{ApA}.

\begin{Lemma}\label{L1}
The PDF describing the distance $r$ between a given TN and the projected location of the AN within the overlapping area of ${\color{black}\mathcal{S}_{\textrm{A}}'} \bigcap {\color{black}\mathcal{S_{\textrm{U}}}}$, conditioned on $M=m_0$, is expressed as: \setcounter{equation}{8}
\begin{align}\label{fr5} 
	f_R(r|m_0)\! =&\! \left\{\!\!\!\! \begin{array}{cl}
		f_{R^{1)}}(r|m_0)\!=\!\frac{2r}{R_{\rm{A}}^2}, \!\!\! &\!\!\!  0<r<R_{\rm{A}} \text{ and} \\
		\!\!\! &\!\!\! 0<m_0\!<\!R_{\rm{U}}\! -\!R_{\rm{A}}, \\
		f_{R^{2)}}(r|m_0)\!=\!\frac{2\pi r}{\gamma}, \!\!\! &\!\!\! 0\! <\! r\! <\! R_{\rm{U}}\! -\! m_0 \text{ and} \\
		\!\!\! &\!\!\! R_{\rm{U}}\! -\! R_{\rm{A}}\!<\!m_0<R_{\rm{U}}, \\
		f_{R^{3)}}(r|m_0), \!\!\! &\!\!\! R_{\rm{U}}\! -\! m_0\! <\! r\! <\! R_{\rm{A}} \text{ and} \\
		\!\!\! & \!\!\! R_{\rm{U}}\! -\! R_{\rm{A}}\!<\! m_0\!<\! R_{\rm{U}}, \\
		f_{R^{3)}}(r|m_0), \!\!\! &\!\!\! m_0\! -\! R_{\rm{U}}\! <\! r\! <\! R_{\rm{A}} \text{ and} \\
		\!\!\! &\!\!\! R_{\rm{U}}\! <\! m_0\!<\!R_{\rm{U}}\! +\! R_{\rm{A}} ,
	\end{array}\!\! \right.\!\!
\end{align}
where $f_{R^{3)}}(r|m_0)$ is given in (\ref{Lemma1.1}) at the bottom of this page,\setcounter{equation}{10}
\begin{align} 
  \gamma =& R_{\rm{U}}^2\! \left(\! \theta_2\! -\! \frac{1}{2}\sin\big(2\theta_2\big)\! \right)\! + R_{\rm{A}}^2\! \left(\! \varphi_2\! -\! \frac{1}{2}\sin\big(2\varphi_2\big)\! \right), \label{eqL1-1} \\
		\theta_2 =& \arccos\left(\frac{m_0 + R_{\rm{U}}^2 - R_{\rm{A}}^2}{2m_0 R_{\rm{U}}}\right) , \label{eqL1-2} \\
	\varphi_2 =& \arccos\left(\frac{m_0 + R_{\rm{A}}^2 - R_{\rm{U}}^2}{2m_0 R_{\rm{A}}}\right) , \label{eqL1-3} \\
	\varphi_3 =& \arccos\left(\frac{m_0^2 - R_{\rm{U}}^2 + r^2}{2m_0r}\right) . \label{eqL1-4}
\end{align}
	
\begin{proof}
See Appendix~\ref{ApA}.
\end{proof}
\end{Lemma}

\begin{remark}\label{Rem1} 
Due to the independent and uniform distribution of TNs within the overlapping area, it can be observed that when the distance between the projection point and the center of the user area is $M=m_0$, the distance variables $r_m$ and $r_n$ for the target TN and interference TNs are independently and identically distributed. Specifically, this can be expressed as $f_{R_m}(r_m|m_0)=f_{R_n}(r_n|m_0)=f_{R}(r|m_0)$.
\end{remark}

\begin{Lemma} \label{L2}
The PDF of the distance $M$ from the projection point of the AN to the center of the user area is given by
\begin{align}\label{fm} 
	f_M(m_0) =& \frac{2m_0}{\big(R_{\rm{A}}+R_{\rm{U}}\big)^2} , ~  0<m_0<R_{\rm{A}}+R_{\rm{U}}.
\end{align}
\end{Lemma}

\begin{proof}
Considering that the movement range of the projection point is within ${\cal B}(\bm{u},R_{\textrm{U}})$, the range of the distance $m_0$ from the projection point to the center of the user area ${\color{black}\mathcal{S_{\textrm{U}}}}$ is $0<m_ 0<R_{\textrm{U}}+R_{\textrm{A}}$. Thus the CDF of $M$ can be expressed as
\begin{align}\label{Fm1} 
	F_M(m_0) =& \frac{m_0^2}{\big(R_{\textrm{A}}+R_{\textrm{U}}\big)^2} , ~ 0<m_0<R_{\textrm{A}}+R_{\textrm{U}}.
\end{align}
The corresponding PDF can be derived by differentiating (\ref{Fm1}) with respect to $m_0$. This completes the proof. 
\end{proof}

\section{Coverage Probability}\label{S4}

In this section, we conduct an analysis on the coverage probabilities for both the T-A link and A-S link, assuming that at least one user is in ${\color{black}\mathcal{S}_{\textrm{A}}'}$. The coverage probability refers to the likelihood that the SINR at the receiver is larger than the minimum SINR threshold necessary for successful data transmission. In other words, if the SINR of the signal at a receiver surpasses a pre-established threshold, the transmitter is deemed to be inside the coverage area of the receiver's communication network.
		
\begin{figure*}[t]	\setcounter{equation}{20}
\begin{equation}\label{Lemma4.1} 
	\mathcal{L}_{I_{\textrm{T}}}(s) = \left\{
		\begin{aligned}
			\mathcal{L}_{I_1}(s) =& \sum_{n_I=1}^{N_0-1} \binom{N_0-1}{n_I} \frac{R_{\textrm{A}}^{2n_I}\big(R_{\textrm{U}}^2-R_{\textrm{A}}^2\big)^{N_0-1-n_I}}{R_{\textrm{U}}^{2(n_I+1)}} \left(\int_{0}^{R_{\textrm{A}}} D_2 \,\frac{2r_n}{R_{\textrm{A}}^2}\mathrm{d}r_n \right)^{\!n_I}\!\! , &\!\! 0 < m_0 <R_{\textrm{U}}\!-\!R_{\textrm{A}} ,\\ 
			\mathcal{L}_{I_2}(s) =& \sum_{n_I=1}^{N_0-1} \binom{N_0-1}{n_I} \frac{\gamma^{n_I}\big(R_{\textrm{U}}^2-\gamma\big)^{N_0-1-n_I}}{R_{\textrm{U}}^{2(n_I+1)}} \\ 
			&\times \left(\int_{R_{\textrm{U}}-m_0}^{R_{\textrm{A}}} D_2 \frac{2\pi r}{\gamma}\mathrm{d}r_n+\int_{R_{\textrm{U}}-m_0}^{R_{\textrm{A}}}D_1 D_2 \mathrm{d}r_n\right)^{n_I}\!\! , &\!\! R_{\textrm{U}}\!-\!R_{\textrm{A}}<m_0<R_{\textrm{U}}, \\
			\mathcal{L}_{I_3}(s) =& \sum_{n_I=1}^{N_0-1} \binom{N_0-1}{n_I} \frac{\gamma^{n_I}(R_{\textrm{U}}^2-\gamma)^{N_0-1-n_I}}{R_{\textrm{U}}^{2(n_I+1)}} \left(\int_{m_0-R_{\textrm{U}}}^{R_{\textrm{A}}} D_1D_2\,\mathrm{d}r_n\right)^{n_I}\!\! , &\!\!  R_{\textrm{U}}<m_0<R_{\textrm{U}}\!+\!R_{\textrm{A}} .
		\end{aligned}
  \right.
\end{equation} 
\hrulefill
\vspace*{-2mm}
\end{figure*}

\begin{figure*}[b]\setcounter{equation}{22}
\vspace*{-1mm}
\hrulefill 
\begin{align}\label{cov1} 
  P_{\textrm{cov}}^{\textrm{T-A}} =& \int \int \sum_{n=1}^{N_{\textrm{TA}}} (-1)^{n+1} \binom{N_{\textrm{TA}}}{n} \mathcal{L}_{I_{\textrm{T}}}(s) f_{R_m}(r_m|m_0)\,f_M(m_0) \,\mathrm{d}r_m \,\mathrm{d}m_0 \nonumber \\
  =& \int_0^{R_{\textrm{U}}-R_{\textrm{A}}} \int_0^{R_{\textrm{A}}} \sum_{n=1}^{N_{\textrm{TA}}} \,\,(-1)^{n+1} \binom{N_{\textrm{TA}}}{n} \,\,\mathcal{L}_{I_1}(s) \,\, f_{R^{1)}}(r_m|m_0)\,\, f_M(m_0)\,\,\mathrm{d}r_m \,\mathrm{d}m_0 \nonumber \\
  & + \int_{R_{\textrm{U}}-R_{\textrm{A}}}^{R_{\textrm{U}}} \Bigg( \int_{0}^{R_{\textrm{U}}-m_0} \sum_{n=1}^{N_{\textrm{TA}}}\,\, (-1)^{n+1} \binom{N_{\textrm{TA}}}{n} \,\,\mathcal{L}_{I_2}(s)\,\, f_{R^{2)}}(r_m|m_0)\,\, \mathrm{d}r_m \nonumber \\
  & + \int_{R_{\textrm{U}}-m_0}^{R_{\textrm{A}}} \sum_{n=1}^{N_{\textrm{TA}}} \,\,(-1)^{n+1} \,\,\binom{N_{\textrm{TA}}}{n}\,\, \mathcal{L}_{I_2}(s)  \,\, f_{R^{3)}}(r_m|m_0) \,\,\mathrm{d}r_m	\Bigg) f_M(m_0) \,\,\mathrm{d}m_0 \nonumber \\
  & + \int_{R_{\textrm{U}}}^{R_{\textrm{U}}+R_{\textrm{A}}} \int_{m_0-R_0}^{R_{\textrm{A}}} \sum_{n=1}^{N_{\textrm{TA}}} \,\,(-1)^{n+1} \binom{N_{\textrm{TA}}}{n} \,\, \mathcal{L}_{I_3}(s) f_{R^{3)}}(r_m|m_0)\,\,f_M(m_0)\,\,\mathrm{d}r_m \,\mathrm{d}m_0 . 
\end{align} 
\end{figure*}

\subsection{T-A Link}\label{S4.1}

\begin{Theorem}\label{The1}
The coverage probability of a TN communicating with an AN within the coverage range of the AN under the Nakagami fading channel is given by\setcounter{equation}{16}
\begin{align}\label{cov} 
	P_{\rm{cov}}^{\rm{T}\textrm{-}\rm{A}} \triangleq& \mathbb{P}( \mathrm{SINR}_1\geq T_{h_1}) \nonumber \\
	=& \int \int \sum_{n=1}^{N_{\rm{TA}}}(-1)^{n+1}\binom{N_{\rm{TA}}}{n}\mathbb{E}_{I}\left[ \exp\left(-s I_{\rm{T}} \right)\right] \nonumber \\
	& \times f_{R_m}(r_m|m_0)\,f_M(m_0) \,\mathrm{d}r_m \,\mathrm{d}m_0,		
\end{align}
where $T_{h_1}$ is the SINR threshold of the T-A link, 
\begin{align}\label{eqLapTra} 
 	\mathbb{E}_{I}\left[\mathrm{exp}\left(-s{I_{\rm{T}}} \right)\right] =& \mathcal{L}_{I_{\rm{T}}}\left(s\right),
\end{align}
is the Laplace transform of the cumulative interference power $I_{\rm{T}}$, $s= \frac{n\eta T_{h_1} (H_{\rm{A}}^2+r_m^2)^{\frac{\alpha_1}{2}}}{P_{\rm{T}}}$ with $\eta =N_{\rm{TA}}(N_{\rm{TA}}!)^{-\frac{1}{N_{\rm{TA}}}}$, and $N_{\rm{TA}}$ is Nakagami fading parameter. 
\end{Theorem}

\begin{proof}
See Appendix~\ref{ApB}.
\end{proof}

\begin{Lemma}\label{L3}
Laplace transform of random variable $I_{\rm{T}}$ is 
\begin{align}\label{eqLapTra2} 
	\!\!\mathcal{L}_{I_{\rm{T}}}\!\left(s\right)\! =&\mathbb{E}_{N_I,R_n}\!\! \Bigg[ \!\prod_{T_n\!\in\Phi_{\rm U}  \! {\backslash} T_m} \underset{D_2}{\underbrace{\left(\!1\!+\!\frac{s P_{\rm{T}}}{N_{\rm{TA}}(H_{\rm{A}}^2\! +\! r_n^2)^{\frac{\alpha_1}{2}}}\!\right)^{\!\!\!\!-N_{\rm{TA}}}}} \Bigg]\! ,\!
\end{align}
where the expectation is over the number of interference users $N_I$ and the interfering users' distances $R_n$.
\end{Lemma}

\begin{proof}
See Appendix~\ref{ApC}.
\end{proof}

The point distribution of the interference users can be described by a BPP. Therefore, the number of interference users $N_I$ follows a binomial distribution with a certain probability of success. The probability of success can be expressed as\setcounter{equation}{19}
\begin{eqnarray}\label{P_I} 
	P_{I}\!=\begin{cases}
		P_{I_1}=\frac{R_{\textrm{A}}^2}{R_{\textrm{U}}^2}	, &  0<m_0<R_{\textrm{U}}-R_{\textrm{A}} ,\\ 
		P_{I_2}=\frac{\gamma}{R_{\textrm{U}}^2}, &  R_{\textrm{U}}-R_{\textrm{A}}<m_0<R_{\textrm{U}}+R_{\textrm{A}},
	\end{cases}
\end{eqnarray}
where $\gamma$ is given in (\ref{eqL1-1}). Noting the PDF $f_{R}(r_n|m_0)$ (\ref{fr5}), we obtain the three expressions of $\mathcal{L}_{I_{\rm{T}}}\!\left(s\right)$ in the three different ranges of $0\! <\! m_0\! <\! R_{\textrm{U}}\!-\!R_{\textrm{A}}$, $R_{\textrm{U}}\!-\!R_{\textrm{A}}\! <\! m_0\! <\! R_{\textrm{U}}$ and $R_{\textrm{U}}\! <\! m_0\! <\! R_{\textrm{U}}\!+\!R_{\textrm{A}}$, which are given in (\ref{Lemma4.1}) at the top of the page, where $N_0$ is the total number of users, $D_1\! =\! f_{R^{3)}}(r_n|m_0)$,  and $D_2=\left(\!1\!+\!\frac{s	P_{\textrm{T}}}{N_{\textrm{TA}}(H_{\rm{A}}^2+r_n^2)^{\frac{\alpha_1}{2}}}\!\right)^{\!\!\!\!-N_{\textrm{TA}}}$.

$\mathcal{L}_{I_1}(s)$ for $0<m_0<R_{\textrm{U}}-R_{\textrm{A}}$ is derived as follows:\setcounter{equation}{21}
\begin{align}\label{eqLI1} 
	\mathcal{L}_{I_1}\!(s) =& \!\ \! \sum_{n_I=1}^{N_0-1}\! \binom{N_0-1}{n_I} (P_{I_1})^{n_I}\! \left(\! \int_{0}^{R_{\textrm{A}}} \!\!D_2 f_{R^{1)}}(r_n|m_0) \mathrm{d}r_n\! \right)^{\!n_I} \nonumber \\
	&\times (1-P_{I_1})^{N_0-1-n_I}	\nonumber \\
	=& \!\! \sum_{n_I=1}^{N_0-1}\! \binom{N_0-1}{n_I}\frac{R_{\textrm{A}}^{2n_I}(R_{\textrm{U}}^2-R_{\textrm{A}}^2)^{N_0-1-n_I}}{R_{\textrm{U}}^{2(n_I+1)}}\nonumber \\
	&\times  \left(\int_{0}^{R_{\textrm{A}}} D_2 \,\frac{2r_n}{R_{\textrm{A}}^2}\mathrm{d}r_n \right)^{\!n_I} .
\end{align}
	$\mathcal{L}_{I_2}(s)$ for $R_{\textrm{U}}-R_{\textrm{A}}<m_0<R_{\textrm{U}}$ is derived as follows:
\begin{eqnarray}
	\mathcal{L}_{I_2}\!\left(s\right)\!\!\!\!&=&\!\!\!\!\!\sum_{n_I=1}^{N_0-1}\!\!\binom{N_0-1}{n_I}(P_{I_2})^{n_I}\!\!\left(\int_{0}^{R_{\textrm{U}}-m_0}\!\!\!\!\!\!\! D_2 \,f_{R_2}(r_n|m_0)\mathrm{d}r_n  \right.\nonumber\\
	&+& \!\!\!\!\! \left. \int_{R_{\textrm{U}}-m_0}^{R_{\textrm{A}}}\!\!\!\!\!D_2 f_{R_3}(r_n|m_0)\mathrm{d}r_n \! \right)^{n_I} \!\!\!\!\!\!(1-P_{I_2})^{N_0-1-n_I}	\nonumber\\
	&=& \!\!\!\!\!\sum_{n_I=1}^{N_0-1}\!\binom{N_0-1}{n_I}\frac{\gamma^{n_I}(R_{\textrm{U}}^2-\gamma)^{N_0-1-n_I}}{R_{\textrm{U}}^{2(n_I+1)}}\nonumber\\
	&\times& \!\!\!\!\left(\int_{R_{\textrm{U}}-m_0}^{R_{\textrm{A}}}\!\!\!\!\!D_2 \frac{2\pi r}{\gamma}\mathrm{d}r_n+\int_{R_{\textrm{U}}-m_0}^{R_{\textrm{A}}}D_1 D_2\mathrm{d}r_n\right)^{n_I}\!\!.
\end{eqnarray}

$\mathcal{L}_{I_3}(s)$ for $R_{\textrm{U}}<m_0<R_{\textrm{U}}+R_{\textrm{A}}$ is derived as follows:
\begin{eqnarray}
	\mathcal{L}_{I_3}\!\left(s\right)\!\!\!\!&=&\!\!\!\!\!\!\sum_{n_I=1}^{N_0-1}\!\!\binom{N_0-1}{n_I}\!(P_{I_2})^{n_I}\!\!\left(\int_{m_0-R_{\textrm{U}}}^{R_{\textrm{A}}} \!\!\!\!\!\!\!D_2 \,f_{R_4}(r_n|m_0)\mathrm{d}r_n\!\right)^{\!n_I}\nonumber\\	&\times&\!\!\!(1-P_{I_2})^{N_0-1-n_I}	\nonumber\\
	&=& \!\!\!\!\!\sum_{n_I=1}^{N_0-1}\!\binom{N_0-1}{n_I}\,\, \frac{\gamma^{n_I}(R_{\textrm{U}}^2-\gamma)^{N_0-1-n_I}}{R_{\textrm{U}}^{2\,(n_I+1)}}\nonumber\\
	&\times& \left(\int_{m_0-R_{\textrm{U}}}^{R_{\textrm{A}}} D_1D_2\,\mathrm{d}r_n\right)^{n_I}.
\end{eqnarray}

By substituting (\ref{fr5}), (\ref{fm}) and (\ref{Lemma4.1}) into (\ref{cov}), we obtain the coverage probability $P_{\textrm{cov}}^{\textrm{T-A}}$ of the terrestrial-aerial link given in (\ref{cov1}) at the bottom of the page.

\begin{figure*}[b]\setcounter{equation}{27}
\vspace*{-4mm}
\hrulefill
\vspace*{-2mm}
\begin{align}\label{eqCPa-s} 
	P_{\textrm{cov}}^{\textrm{A-S}} =&	1 - \sum_{k=0}^{\infty} \frac{\Psi (k)}{(\beta-\delta)^{k+1}} \Gamma (k+1) \sum_{t=0}^{k+1} \binom{k+1}{t} (-1)^t \exp\left(\lambda_{\textrm{A}} S_{\textrm{A}} (M_2 -1) \right) \nonumber \\ 
	=& 1 - \sum_{k=0}^{\infty} \frac{(-1)^{k} \kappa \delta^{k} (1-q)_{k}\,}{k! \, (\beta-\delta)^{k+1}} \sum_{t=0}^{k+1} \binom{k+1}{t} (-1)^t  \nonumber \\ 
	& \times \exp\Bigg(\frac{\pi (R_{\textrm{U}}+R_{\textrm{A}})^2 \left( 1- \exp\big(-\pi D_{\mathrm{min}}^{2} \lambda_1\big) \right)}{\pi D_{\mathrm{min}}^{2}} \Bigg(\frac{(2cq)^q \big( 1 + \frac{2 t c P_n 	G_{\textrm{t}} G_{\textrm{r}} \zeta (\beta-\delta) T_{h_2}}{P_m G_{A_m\!-\!S}}\big)^{q-1}}{\left((2 c q+\Omega )\big(1+\frac{2 t c P_n 	G_{\textrm{t}} G_{\textrm{r}} \zeta (\beta-\delta) T_{h_2}}{P_m G_{A_m\!-\!S}}\big) - \Omega\right)^q} \frac{\theta}{2\pi} \nonumber \\ 
	& + \frac{(2 c q)^q \big(1+\frac{2 t c P_n 	g_{\textrm{t}}  G_{\textrm{r}}\zeta (\beta-\delta) T_{h_2}}{P_m G_{A_m\!-\!S}}\big)^{q-1}}
{\left((2 c q+\Omega ) \big(1+\frac{2 t c P_n 	g_{\textrm{t}}  G_{\textrm{r}} \zeta (\beta-\delta) T_{h_2}}{P_m G_{A_m\!-\!S}}\big) - \Omega\right)^q} \left(1-\frac{\theta}{2\pi}\right) - 1 \Bigg) \Bigg) .
\end{align} 
\vspace*{-1mm}
\end{figure*}
		
\subsection{A-S Link}\label{S4.2}

The target transmitter $A_m$ has the transmit power $P_m$, while the power level for other $A_n$ in $\Phi_{\textrm{A}}$ is $P_n$.
	
\begin{Theorem}\label{T2}
The coverage probability for an arbitrarily located AN under the SR fading channel is given by\setcounter{equation}{23}
\begin{align}\label{eqT2} 
		P_{\rm{cov}}^{\rm{A}\textrm{-}\rm{S}} \triangleq & \mathbb{P}\big(\mathrm{SINR}_2\geq T_{h_2}\big) \nonumber \\
	=& 1 - \sum_{k=0}^{\infty} \frac{\Psi (k)}{(\beta-\delta)^{k+1}}\, \Gamma (k+1)\, \sum_{t=0}^{k+1} \binom{k+1}{t} \nonumber \\
	&	\times (-1)^t \mathbb{E}\left[ \exp\big(-{s}'	I_{\rm{A}} \big) \right] ,		
\end{align}
where $\Psi (k) = \frac{(-1)^{k} \kappa \delta^{k}}{(k!)^{2}}(1-q)_{k}$, ${s}'=\frac{t\zeta(\beta-\delta) T_{h_2} d_{m}^{\alpha_2}}{P_m G_{A_m\!-\!S}}$, $\zeta=(\Gamma(k+2))^{-\frac{1}{k+1}}$, and $q$ is Nakagami fading parameter of the SR fading model.
\end{Theorem}

\begin{proof}
See Appendix~\ref{ApD}.
\end{proof}

To simplify the expression of $P_{\textrm{cov}}^{\textrm{A-S}}$, we first note that the expectation of $\mathbb{E}\left[\exp(-{s}'I_{\textrm{A}})\right]$ is over the the distance $d_m$ between $A_m$ and $S$ as well as the interference, i.e.,
\begin{align}\label{eqESI} 
	\mathbb{E}\left[\exp(-{s}'I_{\textrm{A}})\right] =& \mathbb{E}_{d_m,I} \left[\exp(-{s}' I_{\textrm{A}})\right] .
\end{align}
Compared with the finite-size operational range of ANs, typically spanning only a few kilometers, the spatial separation between the satellite and ANs is considerably larger, extending to hundreds even thousands of kilometers. Hence, the distance disparity between the satellite and ANs is negligible. Therefore, it is postulated that all of the aerial transmitters possess an equal transmission distance to $S$, i.e., $d_m=d_0$. Since $d_m$ is approximately a constant, (\ref{eqESI}) depends on the interference only. Consequently, we have
\begin{align}\label{eqLT-IA} 
  \mathbb{E}\left[\exp(-{s}'I_{\textrm{A}})\right] =& \mathcal{L}_{I_{\textrm{A}}}\left({s}'\right) ,
\end{align}
that is, $\mathbb{E}\left[\exp(-{s}'I_{\textrm{A}})\right]$ is the Laplace transform of the cumulative interference power $I_{\textrm{A}}$.

\begin{Lemma}\label{L4}
Laplace transform of random variable $	I_{\rm{A}} $ is
\begin{align}\label{eqL4}  
	\mathcal{L}_{	I_{\rm{A}} }\left({s}'\right) =& \exp\big(	\lambda_{\rm A} 	S_{\rm{A} } (M_2 -1)\big),
\end{align}
where $\lambda_{\rm A}\! =\! \frac{1-\mathrm{exp}(-\pi D_{\mathrm{min}}^{2}\lambda _1)}{\pi D_{\mathrm{min}}^{2}}$, $	S_{\rm{A}} \! =\! \pi 	R_{\rm{C}}^2$, and $ M_2\! =\! M_1(t_1)\frac{\theta}{2\pi}\! +\! M_1(t_2)(1\!-\!\frac{\theta}{2\pi})$ with $M_1(t_l)\! =\! \frac{(2cq)^q(1+2ct_l)^{q-1}}{[(2cq+\Omega )(1+2ct_l)-\Omega]^q} $, $l\!\in\!\left\{ 1,2\right\}$ , $t_1\! =\! {s}'P_n d_0^{-\alpha_2}	G_{\rm{t}} 	G_{\rm r} $ and  $t_2\! =\! {s}'P_n d_0^{-\alpha_2}	g_{\rm{t}} 	G_{\rm r} $.
\end{Lemma}

\begin{proof}
See Appendix~\ref{ApE}.
\end{proof}

By inserting (\ref{eqLT-IA}) and (\ref{eqL4}) into (\ref{eqT2}), we obtain the closed-form analytical expression of $P_{\textrm{cov}}^{\textrm{A-S}}$ for the A-S link in (\ref{eqCPa-s}) at the bottom of this page.
\section{Average Ergodic Rate}\label{S5}
	
	The average achievable ergodic rate, also known as the Shannon throughput, is measured in bits per second per Hertz (bits/s/Hz). It represents the mean data rate that can be achieved by a communication system using a normalized frequency resource of 1 Hz. This metric also corresponds to the ergodic capacity of a fading communication link, normalized to a unit bandwidth. Formally, the average ergodic rate is defined as 
\setcounter{equation}{28}
\begin{eqnarray}\label{eqAER}	
	\bar{C} \triangleq \frac{1}{K}\mathbb{E}\left[ \log_2\left(1+\text{SINR}\right) \right] .
\end{eqnarray}

\subsection{T-A Link}\label{S5.1}

\begin{Theorem}\label{T3}
	Under a Nakagami fading  channel, the average rate of any arbitrary terrestrial node is given by
	\begin{align}\label{eqT3} 
	\bar{C}^{\rm{T}\textrm{-}\rm{A}} 	&= \frac{1}{K}\int \!\!\int \!\! \int_{t>0}  \sum_{n=1}^{N_{\rm{TA}}}(-1)^{n+1}\binom{N_{\rm{TA}}}{n}\mathbb{E}_{I_{\rm{T}}}\! \! \left[\exp\left(-s_1 {I_{\rm{T}}}\right)\right]\! \nonumber\\
		&\hspace*{4mm}\times f_{R_{m}}(r_m|m_0) f_M(m_0)\,\mathrm{d}t\,\mathrm{d}r_{m}\,\mathrm{d}m_{0}, 
	\end{align}
	where $s_1\! =\! \frac{n \eta (2^t\! -\! 1) } {P_{\rm{T}} (H_{\rm{A}}^2+r_m^2)^{-\frac{\alpha_1}{2}}} $,  $\mathbb{E}_{I_{\rm{T}}}\! \left[\exp\left(-s_1 {I_{\rm{T}}}\right)\right]$ is obtained by  replacing $s$ in $\mathcal{L}_{	I_{\rm{T}} }\left({s}\right)$ of (\ref{eqLapTra2}) with $s_1$, $f_{R_{m}}(r_m|m_0)$ is shown as (\ref{fr5}), and $f_M(m_0)$ is shown as (\ref{fm}).
\end{Theorem}

\begin{proof}
	See Appendix~\ref{ApF}.
\end{proof}

\subsection{A-S Link} \label{S5.2}

\begin{Theorem}\label{T4}
	Under an SR fading channel, the average rate of any arbitrary aerial node is given by
	\begin{align}\label{eqT3-m} 
		& \bar{C}^{\rm{A}\textrm{-}\rm{S}}\!\! =\!\!  \frac{1}{K}\int_{t>0}\! \Bigg(1 - \sum_{k=0}^{\infty} \frac{\Psi (k)}{(\beta-\delta)^{k+1}} \Gamma (k+1) \sum_{v=0}^{k+1} \binom{k+1}{v} \nonumber \\
		&	\hspace*{5mm}\times (-1)^v \mathbb{E}_{I_{\rm{A}}}\left[\exp\big(-{s_1}'I_{\textrm{A}}\big)\right] \Bigg)\!\mathrm{d}t,   
		\end{align}
	with ${s}_1'\! =\! \frac{v\,\zeta(\beta-\delta)\, (2^t-1) \,d_{m}^{\alpha_2}}{P_m G_{A_m\!-\!S}}$, and $\mathbb{E}_{I_{\rm{A}}}\! \! \left[\exp\left(-s_1' {I_{\rm{A}}}\right)\right]$ is obtained by replacing $s'$ in $\mathcal{L}_{	I_{\rm{A}} }\left({s}'\right)$ of  (\ref{eqL4}) with $s_1'$.
\end{Theorem}

\begin{proof}
	See Appendix~\ref{ApG}.
\end{proof}

\section{Numerical Results}\label{S6}

In this section, we verify the derived analytical expressions using Monte Carlo simulations with 50,000 runs. The results from the analytical derivations of Section~\ref{S4} are indicated in the following figures as `Analysis', while the Monte Carlo results are indicated in the figures as `Simulation'. Unless otherwise specifically stated, the default system parameters utilized in the simulations are listed in Table~\ref{Table2}. 

\begin{table}[t] \scriptsize
	\vspace*{-1mm}
\caption{Default Simulation System Parameters.} 
\label{Table2}
\vspace*{-4mm}
\begin{center}
	\begin{tabular} {c|c|c}
		\Xhline{1.2pt}
		\toprule
		\textbf{Notation} & \textbf{Parameters}& \textbf{Values} \\
		\midrule
		$H_{\textrm{A}}$ & Height of ANs & 0.05\,km \\
		$d_0$	& Distance between ANs and the satellite & 400\,km \\
		$R_{\textrm{U}}$, $R_{\textrm{A}}$ & Radius of ${\color{black}\mathcal{S_{\textrm{U}}}}$ and ${\color{black}\mathcal{S_{\textrm{A}}}'}$ & 9.5\,km, 0.5\,km \\
		\emph{$D_{\mathrm{min}}$} & \makecell{Minimum distance between every\\ two candidate points} & 1\,km \\
		$	P_{\textrm{T}}$, $	P_{\textrm{A}}$ & Power of transmitters in \emph{TN} and \emph{AN} & 20\,dBw, 20\,dBw \\
		$	G_{\textrm{t}}$, $g_{\textrm{t}}$ & Antenna gains of the main and side lobes & 10\,dB, -10\,dB \\
		$	\lambda_{\textrm{T}}$ & Density of terrestrial nodes & $10^{-4}$ \\
		$\lambda_{1}$ & Density of candidate points & $5\times10^{-7}$ \\
		$N_{\textrm{TA}}$ & Nakagami fading parameter & 3 \\
		$(c,q,\Omega)$ & SR fading model & $(0.158,1,0.1)$ \\
		$\alpha_1$, $\alpha_2$ & Path-loss exponents of T-A and A-S link & 2, 2 \\
		$T_{h_1}$ & SINR threshold of the T-A link & variable \\
		$T_{h_2}$ & SINR threshold of the A-S link & -20\,dB \\
	 \bottomrule
\end{tabular}	
\end{center}
\vspace*{-4mm}
\end{table}

\begin{figure}[!t]
	\vspace*{-1mm}
	\begin{center}
		\includegraphics[width=0.9\columnwidth]{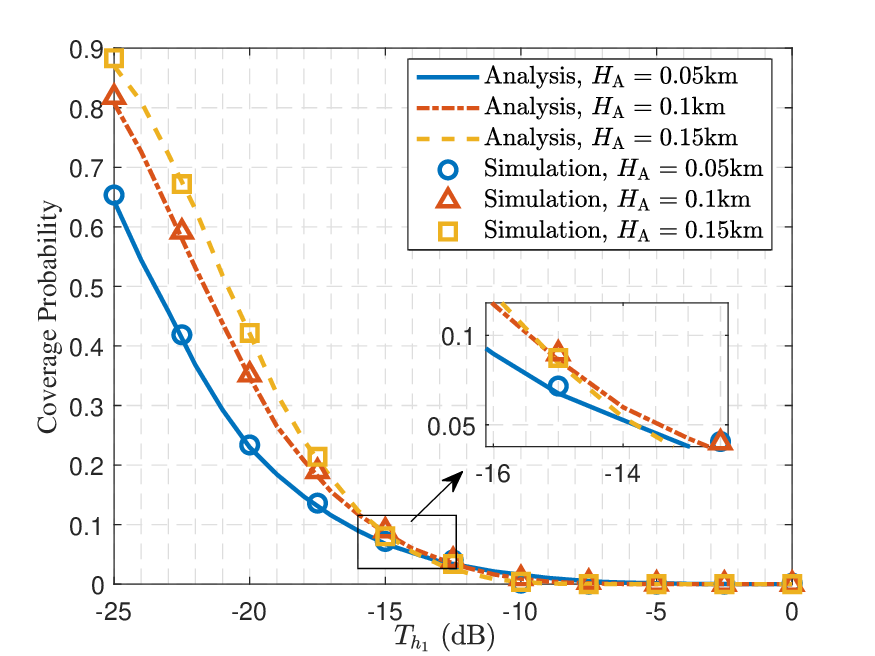}
	\end{center}
	\vspace*{-6mm}
	\caption{Coverage probability of the terrestrial-aerial link as the function of SINR threshold $T_{h_1}$ given three different $H_{\textrm{A}}$.}
	\label{fig:H} 
	\vspace*{-6mm}
\end{figure}
\subsection{Performance of Terrestrial-Aerial Link}\label{S6.1}

In this subsection, we conduct a simulation under various system parameters to verify the analytical coverage probability and average ergodic rate of the terrestrial-aerial link derived in Subsection~\ref{S4.1}, and the results obtained are  depicted in Figs.~\ref{fig:H} to \ref{fig:Rate2}. It can be seen that the analytical results closely match the corresponding simulation results, which supports the validity of our \textbf{Theorem~\ref{The1}}.

\begin{figure}[b]
	\vspace*{-6mm}
	\begin{center}
		\includegraphics[width=0.9\columnwidth]{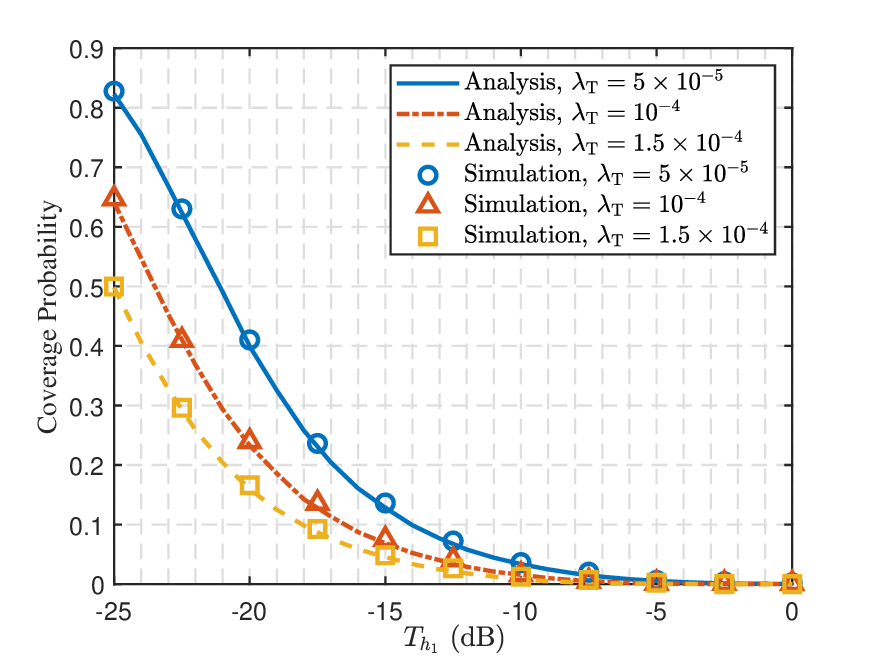}
	\end{center}
	\vspace*{-6mm}
	\caption{Coverage probability of the terrestrial-aerial link as the function of SINR threshold $T_{h_1}$ given three different ${\color{black}\lambda_{\textrm{T}}}$.}
	\label{fig:LambdaT} 
	\vspace*{-1mm}
\end{figure}
Fig.~\ref{fig:H} depicts the coverage probability as the function of the SINR threshold $T_{h_1}$, given three different values of the height of the ANs $H_{\textrm{A}}$. It can be seen from Fig.~\ref{fig:H} that increasing the SINR threshold $T_{h_1}$ decreases the coverage probability, i.e., decreasing the likelihood of experiencing the link coverage. This is expected due to the inverse relationship between $T_{h_1}$ and the probability of achieving an SINR that surpasses the given threshold value. Notably, the results of Fig.~\ref{fig:H} indicates that increasing the height of ANs enhances the coverage probability of the T-A link. This can be explained by examining the impact of $H_{\textrm{A}}$ on the SINR. From the SINR (\ref{eqSINR}) of the T-A link and the accumulative interference power (\ref{eqInterf}) of the link, it is clear that the reduction in the MUI is far more than the reduction in the target signal power when increasing $H_{\textrm{A}}$. Therefore, increasing $H_{\textrm{A}}$ increases the link SINR, leading to the enhancement of the coverage probability.

\begin{figure}[!t]
\vspace*{-1mm}
\begin{center}  
\includegraphics[width=0.9\columnwidth]{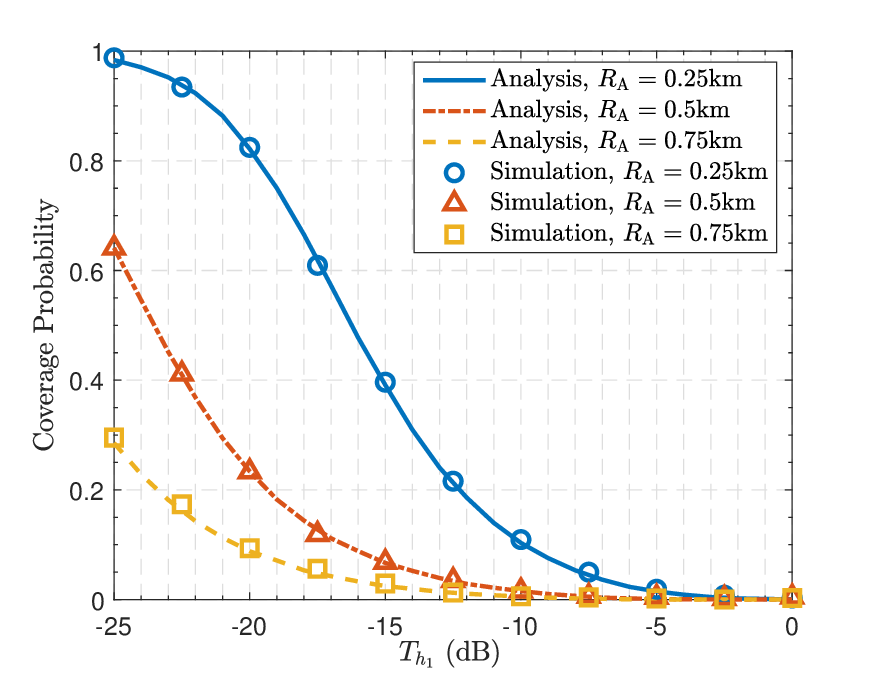}
\end{center}
\vspace*{-6mm}
\caption{Coverage probability of the terrestrial-aerial link as the function of SINR threshold $T_{h_1}$ given three different $R_{\textrm{A}}$.}
\label{fig:R_C} 
\vspace*{-3mm}
\end{figure}

Fig.~\ref{fig:LambdaT} investigates the impact of the density of terrestrial nodes $\lambda_{\textrm{T}}$ on the achievable coverage probability performance. It can be seen from Fig.~\ref{fig:LambdaT} that increasing the  density of terrestrial nodes reduces the achievable coverage probability. This is because a higher $	\lambda_{\textrm{T}}$ indicates a higher number of terrestrial transmitters concurrently attempting to access the UAV, leading to a higher MUI and consequently a lower coverage probability.
 
\begin{figure}[!b]
	\vspace*{-7mm}
	\begin{center}	
		\includegraphics[width=0.9\columnwidth]{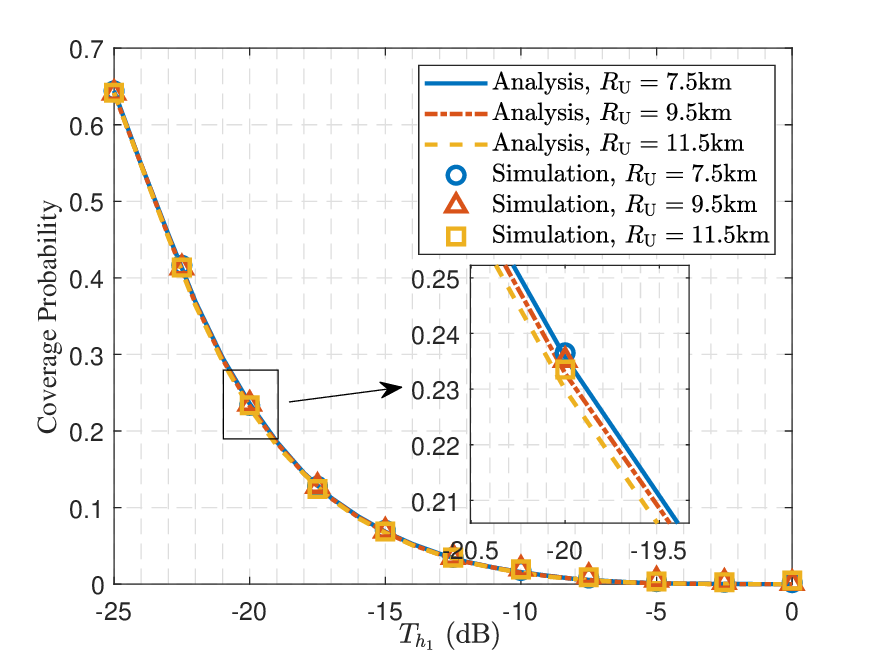}
	\end{center}
	\vspace*{-6mm}
	\caption{Coverage probability of the terrestrial-aerial link as the function of SINR threshold $T_{h_1}$ given three different $R_{\textrm{U}}$.}
	\label{fig:R_U} 
	\vspace*{-1mm}
\end{figure} 

Fig.~\ref{fig:R_C} portrays the coverage probability as the function of the SINR threshold $T_{h_1}$, given three different values of the AN's ground coverage radius $R_{\textrm{A}}$, where the influence of $R_{\textrm{A}}$ on the achievable coverage probability is clearly exhibited. Specifically, increasing the AN's ground coverage radius leads to noticeably reduction in the coverage probability. This is because a larger ground region ${\color{black}\mathcal{S_{\textrm{A}}}'}$ covers more terrestrial nodes, which can communicate with the same AN. This results in a higher number of territorial transmitters concurrently attempting to access the AN, leading to a higher MUI and consequently a lower coverage probability.

 \begin{figure}[!t]  
	\vspace*{-1mm}
	\begin{center}
		\includegraphics[width=0.9\columnwidth]{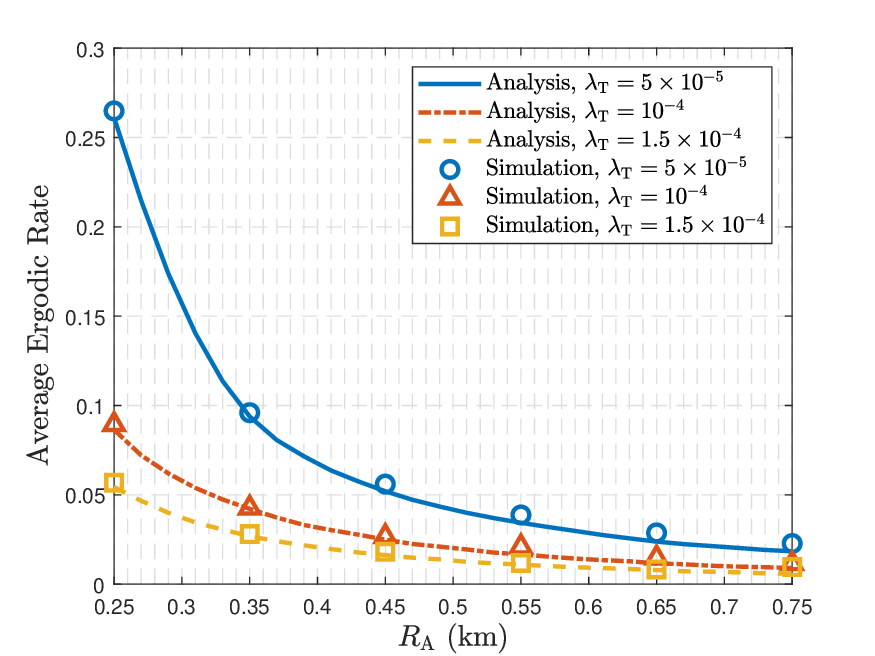}
	\end{center}
	\vspace*{-6mm}
	\caption{Average ergodic rate of the terrestrial-aerial link as the function of AN's ground coverage radius $R_{\textrm{A}}$ given three different values of $\lambda_{\textrm{T}}$.}
	\label{fig:Rate1} 
		\vspace*{-6mm}
\end{figure}

Fig.~\ref{fig:R_U} studies the influence of the radius $R_{\textrm{U}}$ on the coverage probability, indicating that impact of $R_{\textrm{U}}$ on the coverage probability is negligible. Increasing $R_{\textrm{U}}$ increases the area ${\color{black}\mathcal{S_{\textrm{U}}}}$ of terrestrial nodes but this hardly changes the number of the TNs within the AN's coverage area ${\color{black}\mathcal{S_{\textrm{A}}}'}$, given the same user density $	\lambda_{\textrm{T}}$. That is, the number of ground nodes connected to the same AN is hardly changed. Consequently, the MUI of the T-A link is hardly changed and the coverage probability is hardly affected, when $R_{\textrm{U}}$ is changed.

Fig.~\ref{fig:Rate1} shows the average ergodic rate as the function of AN's ground coverage radius $R_{\textrm{A}}$, given three different values for the density of terrestrial nodes $\lambda_{\textrm{T}}$. As expected, increasing $R_{\textrm{A}}$ enhances the average ergodic rate. This is because a larger coverage radius results in more interference nodes, given a fixed density. Additionally, a higher density of terrestrial nodes $\lambda_{\textrm{T}}$ leads to more ground nodes, which increases the MUI of the T-A link, thereby reducing the average ergodic rate.

 \begin{figure}[!b]  
	\vspace*{-6mm}
	\begin{center}
		\includegraphics[width=0.9\columnwidth]{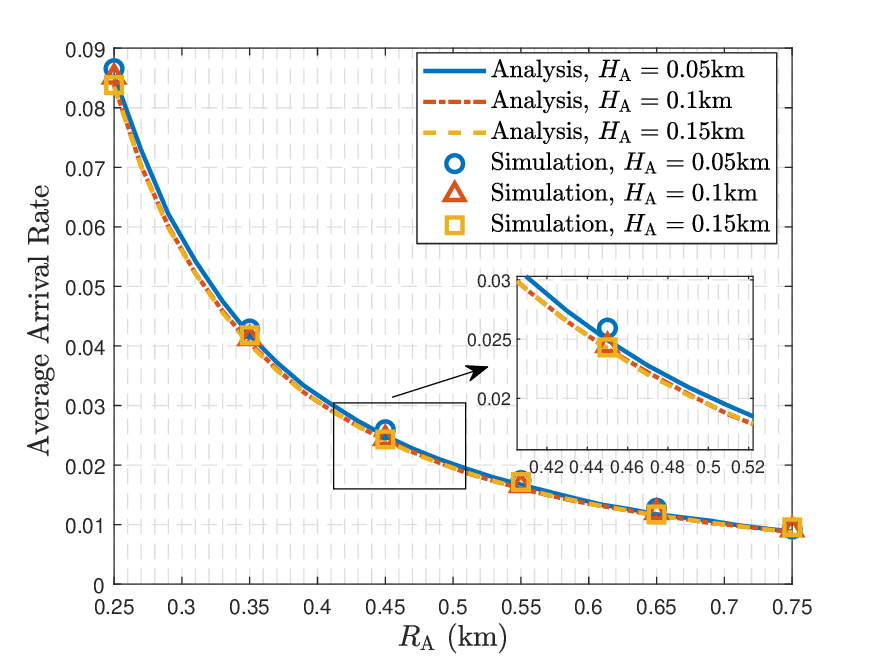}
	\end{center}
	\vspace*{-6mm}
	\caption{Average ergodic rate of the terrestrial-aerial link as the function of AN's ground coverage radius $R_{\textrm{A}}$ given three different values of $H_{\textrm{A}}$.}
	\label{fig:Rate2} 
	\vspace*{-1mm}
\end{figure}

  \begin{figure}[!t]  
	\vspace*{-1mm}
	\begin{center}
		\includegraphics[width=0.9\columnwidth]{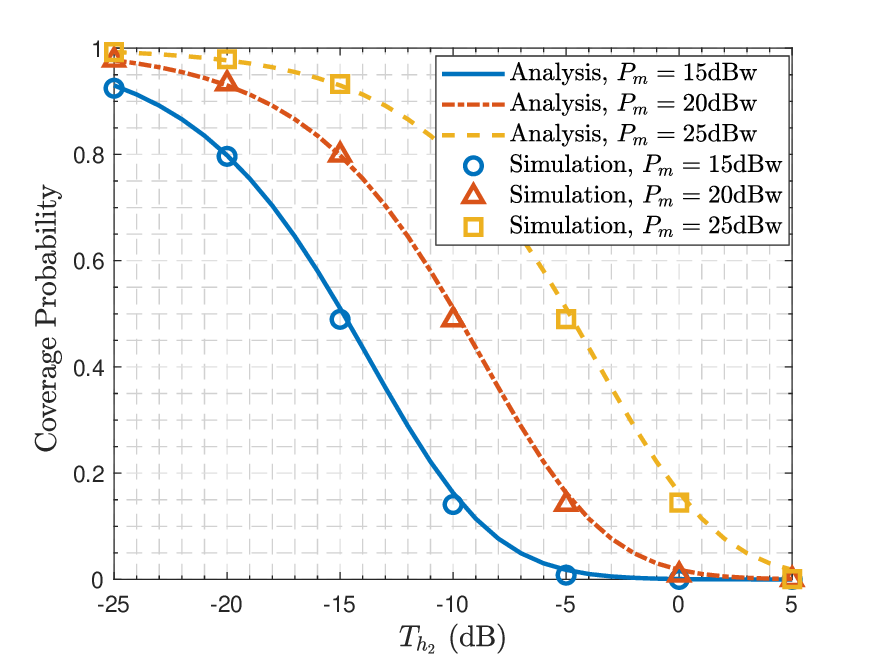}
	\end{center}
	\vspace*{-6mm}
	\caption{Coverage probability of the aerial-satellite link as the function of the SINR threshold $T_{h_2}$ given three different values of $P_m$.}
	\label{fig:Pm} 
	\vspace*{-7mm}
\end{figure}
Fig.~\ref{fig:Rate2} demonstrates the impact of the values of the height of the ANs, $H_{\textrm{A}}$, on the data transmission rate of the T-A link. As explained for Fig.~\ref{fig:Rate2}, an increase of the values of  $H_{\textrm{A}}$ reduces the average ergodic rate slightly. This can be explained by examining the impact of $H_{\textrm{A}}$ on the SINR. However, the increase of $H_{\textrm{A}}$ has a negligible impact on the SINR, resulting in only minor variations in the average ergodic rate.

\subsection{Performance of Aerial-Satellite Link}\label{S6.2}

Similarly, Monte Carlo simulations are employed to validate the close-form analytical coverage probability of the aerial-satellite link provided by \textbf{Theorem~\ref{T2}}. In the simulation, all the interfering ANs have the same transmit power of $P_n = 	P_{\textrm{A}}$, while the target AN's transmit power $P_m$ is a variable. The results obtained are presented in Figs.~\ref{fig:Pm} to \ref{fig:ratelambda}, which confirm that Monte Carlo simulated coverage probability closely matches the analytical theoretical result.

 \begin{figure}[!b]
	\vspace*{-7mm}
	\begin{center}
		\includegraphics[width=0.9\columnwidth]{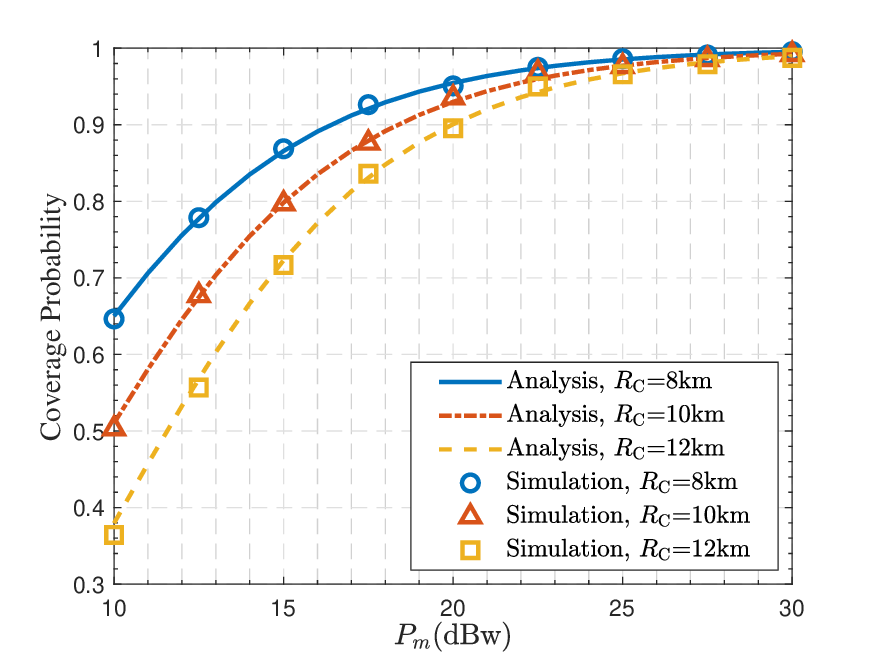}
	\end{center}
	\vspace*{-6mm}
	\caption{Coverage probability of the aerial-satellite link as the function of target AN's transmit power $P_m$ given three different values of $R_{\textrm{C}}$.}
	\label{fig:RC} 
	\vspace*{-1mm}
\end{figure}

  More specifically, in Fig.~\ref{fig:Pm}, we analyze the relationship between the coverage probability and the SINR threshold. From Fig.~\ref{fig:Pm}, it is evident that there exists an inverse relationship between the threshold value $T_{h_2}$ and the coverage probability, whereby an increase of $T_{h_2}$ leads to a decrease of the coverage probability. Furthermore, given the other network parameters unchanged, an increase of the target AN's transmit power $P_m$ leads to an enhanced coverage probability. This is because when the transmit signal strength of the targeted aerial node is enhanced while the interference power from other aerial nodes remain unchanged, the SINR increases, which manifests an increase in coverage probability.

\begin{figure}[!t]
	\vspace*{-1mm}
	\begin{center}
		\includegraphics[width=0.9\columnwidth]{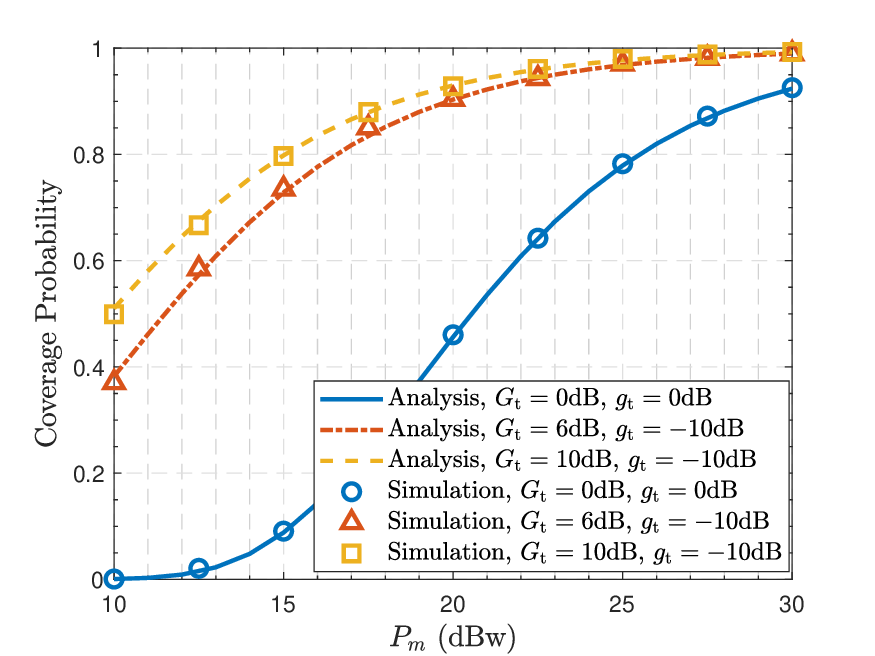}
	\end{center}
	\vspace*{-6mm}
	\caption{Coverage probability of the aerial-satellite link as the function of target AN's transmit power $P_m$ given three different combinations of  $	G_{\textrm{t}}$ and $g_{\textrm{t}}$.}
	\label{fig:G} 
	\vspace*{-7mm}
\end{figure}

Fig.~\ref{fig:RC} plots the coverage probability of the A-S link as the function of target AN's transmit power $P_m$, given three different values for the radius $R_{\textrm{C}}$ of the circle space ${\color{black}\mathcal{S_{\textrm{A}}}}$. As expected, increasing $P_m$ increases the coverage probability, since increasing the power of the target transmitter leads to an increase in its SINR and this reduces the risk of communication interruption. In addition, the impact of $R_{\textrm{C}}$ on the coverage probability is clearly shown in Fig.~\ref{fig:RC}, namely, the expansion of the distribution space of aerial transmitters leads to a worsen coverage probability. This is because the increased availability of space for the MHCPP deployment of air nodes results in a higher number of transmitters concurrently attempting to access the satellite, leading to a higher MUI and consequently a worsen  coverage probability.

\begin{figure}[!b]
\vspace*{-7mm}
\begin{center}
\includegraphics[width=0.9\columnwidth]{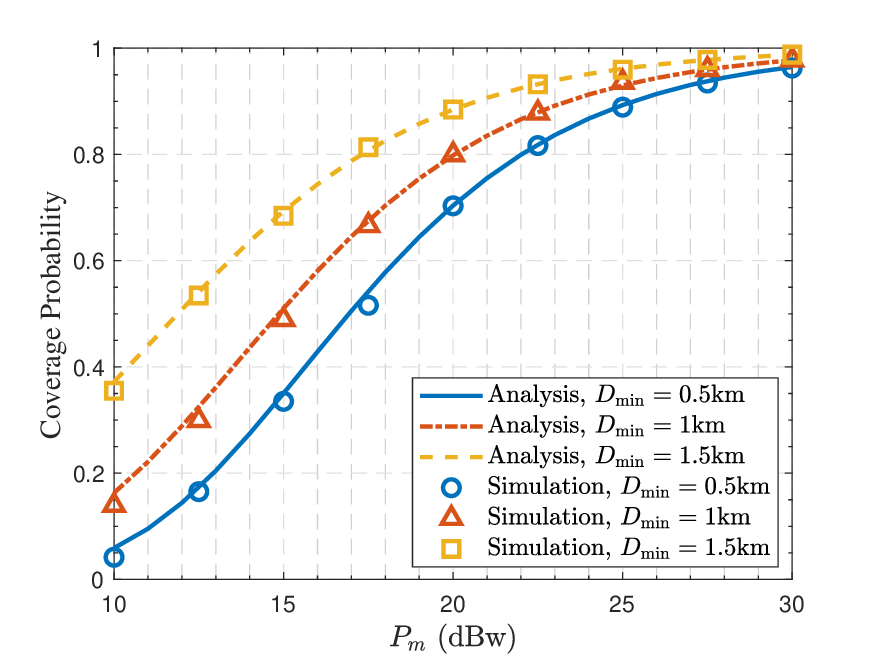}
\end{center}
\vspace*{-6mm}
\caption{Coverage probability of the aerial-satellite link as the function of target AN's transmit power $P_m$ given three different values of $D_{\textrm{min}}$.}
\label{fig:Dmin} 
\vspace*{-1mm}
\end{figure}

Fig.~\ref{fig:G} investigates the impact of different antenna gains on the coverage probability, and the results indicate that higher mainlobe gain and lower sidelobe gain improve the achievable coverage probability performance. The case of both $	G_{\textrm{t}}$ and $g_{\textrm{t}}$ being 0\,dB simulates the scenario without directional BF, while the other two sets of $	G_{\textrm{t}}$ and $g_{\textrm{t}}$ represent scenarios with different directional BF gains. As expected, the utilization of directional BF yields a notable enhancement in coverage probability, compared with the case of without directional BF. With directional BF, the satellite is in the mainlobe of the target transmitter, while it is randomly in the mainlobe or sidelobe of interfering transmitters. Furthermore, a higher disparity in the strength between the mainlobe and side lobes results in a reduced overall interference on the intended recipient. Consequently, the SINR of the signal received from the target node at the satellite is enhanced, leading to an improvement in the coverage probability.

Fig.~\ref{fig:Dmin} studies the impact of the minimum distance $D_{\textrm{min}}$ between every two candidate points on the coverage probability. Evidently, higher $D_{\textrm{min}}$ leads to lower aerial transmitters' access to the satellite, thereby resulting in a decrease of the interference towards the intended transmitter. Consequently, the SINR of the target link increases, which enhances the likelihood of coverage.

\begin{figure}[!t]  
	\vspace*{-1mm}
	\begin{center}
		\includegraphics[width=0.9\columnwidth]{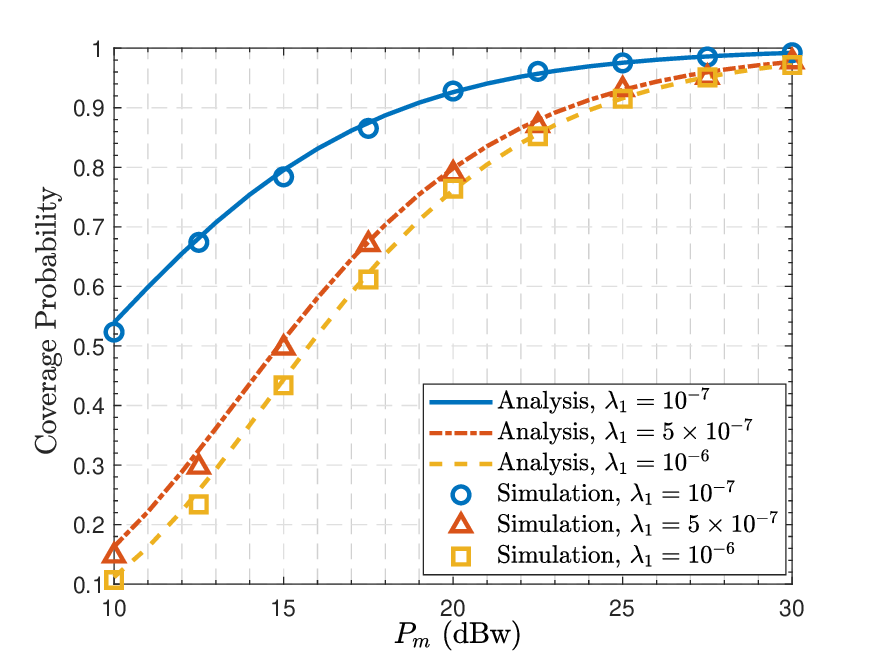}
	\end{center}
	\vspace*{-6mm}
	\caption{Coverage probability of the aerial-satellite link as the function of target AN's transmit power $P_m$ given three different values of $\lambda_{1}$.}
	\label{fig:lambda1} 
	\vspace*{-6mm}
\end{figure}
\begin{figure}[!b]  
	\vspace*{-6mm}
	\begin{center}
		\includegraphics[width=0.9\columnwidth]{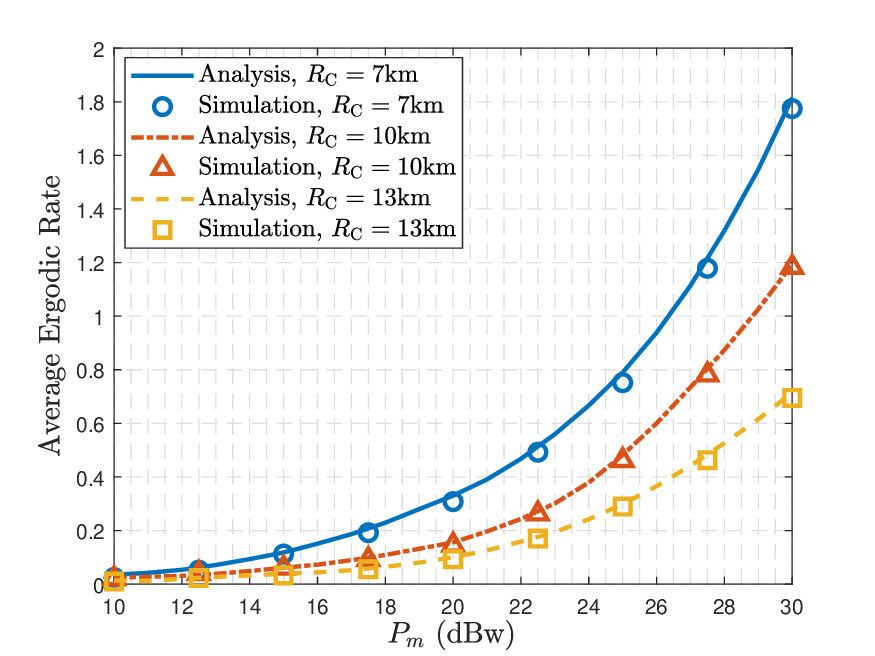}
	\end{center}
	\vspace*{-6mm}
	\caption{Average ergodic rate of the aerial-satellite link as the function of target AN's transmit power $P_m$ given three different values of $R_{\textrm{C}}$.}
	\label{fig:rateRC} 
	\vspace*{-1mm}
\end{figure}

In Fig.~\ref{fig:lambda1}, we proceed to compare the coverage probability under various values of $\lambda_1$.   The graph illustrates a negative correlation between the node density and the coverage probability. It can be easily comprehended that an augmented density corresponds to an increased number of nodes, consequently resulting in a greater MUI. This, in turn, contributes to a heightened risk of outage. It may be readily understood that an enhanced density is associated with a higher quantity of nodes, hence leading to an amplified MUI. Consequently, this phenomenon leads to an increased susceptibility to service disruption, manifested as a decrease in the coverage probability.

Fig.~\ref{fig:rateRC} depicts the average ergodic rate as the function of the target AN's transmit power $P_m$, given three different values for the radius $R_{\textrm{C}}$ of the ANs' deployment area  ${\color{black}\mathcal{S_{\textrm{A}}}}$. As expected, increasing $P_m$ increases the average ergodic rate. In addition, the impact of $R_{\textrm{C}}$ on the achievable average ergodic rate is clearly shown in Fig.~\ref{fig:rateRC}. Evidently, increasing the available area of ${\color{black}\mathcal{S_{\textrm{A}}}}$ results in a reduction of the average ergodic rate, because there are more interfering transmitters. Observe that Fig. \ref{fig:rateRC} is consistent with Fig.~\ref{fig:RC}, which is to be expected given the relationship between the coverage probability and the average ergodic rate.

\begin{figure}[!t]  
		\vspace*{-2mm}
	\begin{center}
		\includegraphics[width=0.9\columnwidth]{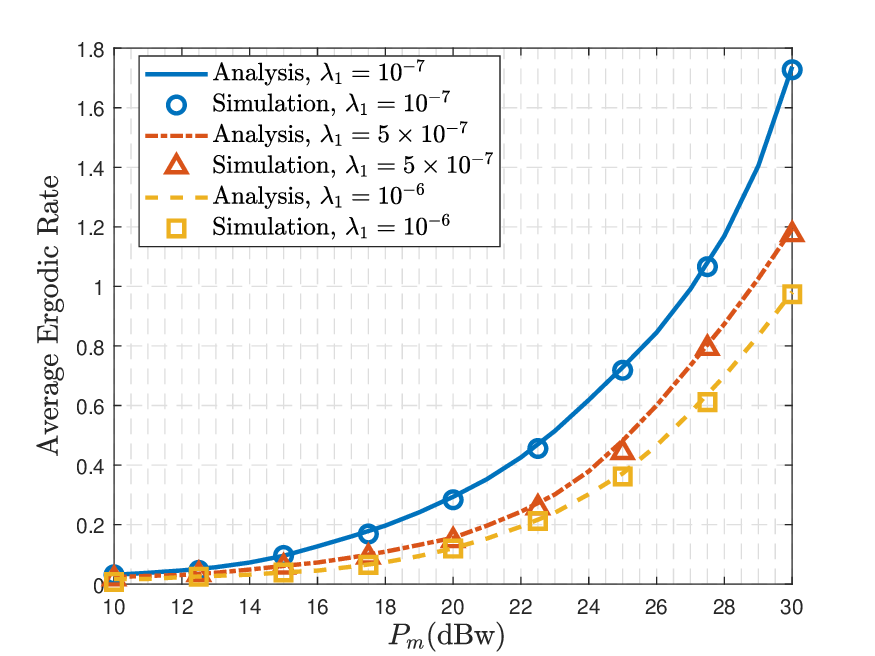}
	\end{center}
	\vspace*{-6mm}
	\caption{Average ergodic rate of the aerial-satellite link as the function of target AN's transmit power $P_m$ given three different values of $\lambda_{1}$.}
	\label{fig:ratelambda} 
	\vspace*{-7mm}
\end{figure}

Fig.~\ref{fig:ratelambda} further investigates the impact of the node density on the average ergodic rate. As expected, increasing $\lambda_{1}$ reduces the average ergodic rate, since increasing $\lambda_{1}$ leads to more interfering nodes. Obviously, Fig.~\ref{fig:ratelambda} is entirely consistent with Fig.~\ref{fig:lambda1}, because a higher coverage probability corresponds to a higher average ergodic rate and vice versa.

\section{Conclusions}\label{S7}

In this paper, we have proposed a tractable approach for analyzing the coverage probability   and the average ergodic rate  of T-A links and A-S links in a CSATN system,  whose terrestrial terminals are located in a finite-size region. This condition incurs significant challenge for the performance analysis.   Utilizing the expressions of coverage probability derived under various conditions, we can input relevant parameter values in analogous scenarios to determine the joint coverage probability of the end-to-end links spanning the terrestrial terminals, the aerial relays, and the satellite. Furthermore, with these theoretical results, we can gain a clear understanding of the impact imposed by critical system parameters, including the coverage area of aerial nodes, the flying altitude of aerial nodes, as well as the terrestrial and aerial nodes’ densities, transmission distance, and antenna gain, on the achievable system performance. Therefore, our study offers theoretical guidance and valuable insights on how to conduct CSATN planning, deployment and optimization in practice. Our future work will focus on analyzing the uplink ergodic sum rate for the CSATNs studied here.

\appendix
\begin{figure*}[!t]
	\begin{minipage}[t]{0.28\linewidth} %
		\centering
		\subfigure[$0\!<\!m_0\!<\!R_{\textrm{U}}\!-\!R_{\textrm{A}} \!\!$] {	
			\centering
			\includegraphics[width=0.96\linewidth]{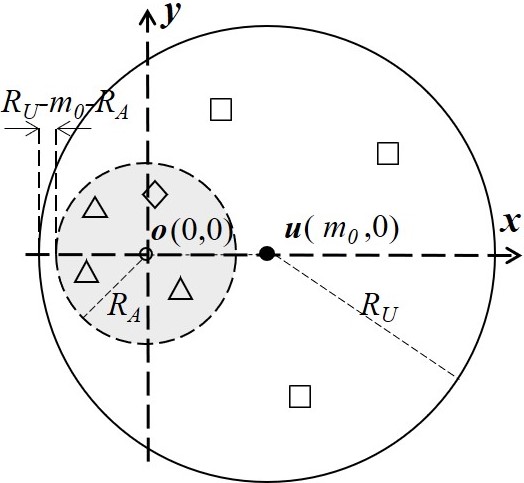}
			\label{fr1}
		}
	\end{minipage}  
	\begin{minipage}[t]{0.28\linewidth} %
		\centering
		\subfigure[$R_{\textrm{U}}\!\!-\!\!R_{\textrm{A}}\!<\!m_0\!<\!R_{\textrm{U}}$] {	
			\centering
			\includegraphics[width=1.02\linewidth]{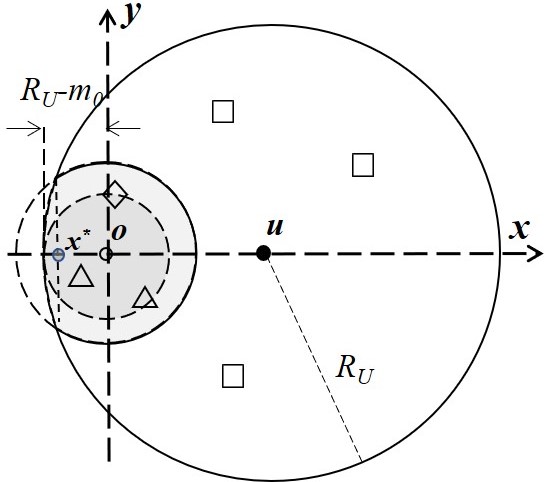}
			\label{fr2}
		}
	\end{minipage}
	\begin{minipage}[t]{0.28\linewidth} %
		\centering
		\subfigure[ $R_{\textrm{U}}\!<\!m_0\!<\!R_{\textrm{U}}\!+\!R_{\textrm{A}} \!\!$ ] {	
			\centering
			\includegraphics[width=1.48\linewidth]{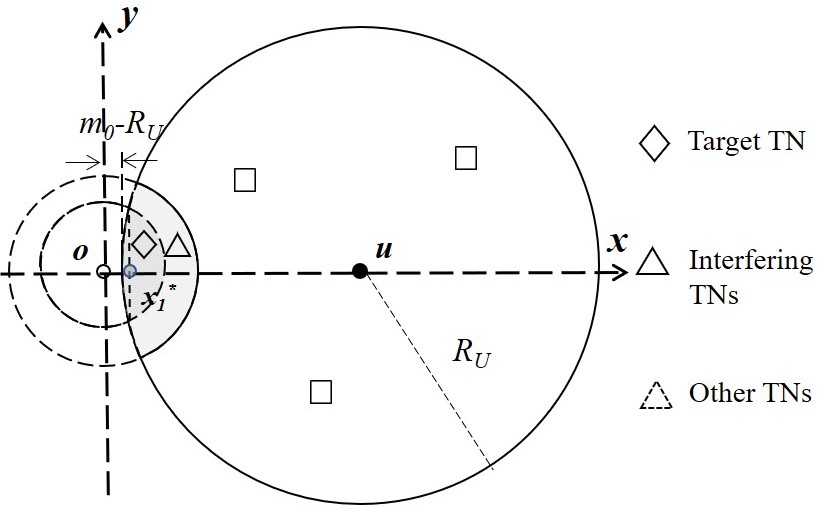}
			\label{fr4}
		}
	\end{minipage}
	\vspace*{-2mm}
	\caption{Illustration of positional relationships under different conditions.}
	\label{fr} 
	\vspace*{-1mm}
\end{figure*}

\subsection{Proof of Lemma~\ref{L1}}\label{ApA}

\begin{proof}
In order to facilitate calculations, we provide a two-dimensional Cartesian coordinate system, as depicted in Fig.~\ref{fr}. The AN projection point, denoted as $\bm{o}$, is situated at the origin of the coordinate system. The center point $\bm{u}$ of the user area is situated in the positive half of the x-axis. Additionally, conditioned on $M=m_0$, the coordinates of $\bm{u}$ is given by $(m_0,0)$. 

\emph{1)}~Given that $0<m_0<R_{\textrm{U}}-R_{\textrm{A}}$, ${\color{black}\mathcal{S}_{\textrm{A}}'}$ is entirely contained within ${\color{black}\mathcal{S_{\textrm{U}}}}$, and $0<r<R_{\textrm{A}}$. The CDF of $r$ is \setcounter{equation}{31}
\begin{align}\label{eqApA1} 
  F_{R^{1)}}(r|m_0) =& \frac{\pi r^2}{\pi R_{\textrm{A}}^2} = \frac{r^2}{R_{\textrm{A}}^2},
\end{align}
and the corresponding PDF is given by
\begin{align}\label{eqApA2} 
  f_{R^{1)}}(r|m_0) =& \frac{2r}{{\color{black}R_{\textrm{A}}}^2}.
\end{align}

Given that $R_{\textrm{U}}-R_{\textrm{A}}<m_0<R_{\textrm{U}}$, ${\color{black}\mathcal{S}_{\textrm{A}}'}$ and ${\color{black}\mathcal{S_{\textrm{U}}}}$ exhibit partial overlap. It is possible to further categorize this scenario into two distinct situations, dependent on whether the projection point falls inside the area of ${\cal S}_{\cal U}$. According to Fig.~\ref{fr2}, it is evident that the region of overlap between ${\color{black}\mathcal{S}_{\textrm{A}}'}$ and ${\color{black}\mathcal{S_{\textrm{U}}}}$ corresponds to the intersection area of two circles: one with a radius of $R_{\textrm{A}}$ centered at the origin, and the other with a radius of $R_{\textrm{U}}$ centered at $(m_0,0)$. Hence, the abscissa $x^{*}$ at which the two circles intersect can be expressed as 
\begin{align}\label{eqApA3} 
  x^{*} =& \frac{m_{0}^{2}-R_{\textrm{U}}^{2}+R_{\textrm{A}}^{2}}{2m_{0}}.
\end{align}
Then the intersecting area, denoted as $\gamma$, can be written as
\begin{align}\label{eqApA4} 
	\gamma\! =&\! \int_{\!m_{0}-R_{U}}^{x^{*}}\!\!\!\!\!\!\!\!2\sqrt{R^{2}-(x\! -\! m_{0})^{2}}{\rm d}x+\! \int_{x^{*}}^{R_{A}}\!\!\!2\sqrt{R_{A}^{2}-x^{2}}{\rm d}x.
\end{align}
By setting $\frac{m_0-x}{R_{\textrm{U}}}=\cos (\theta)$ and $\frac{x}{R_{\textrm{A}}}=\cos (\varphi)$, we obtain
\begin{align} \label{gamma} 
	\gamma =& R_{\textrm{U}}^2\! \left(\! \theta_2 - \frac{1}{2}\sin\big(2\theta_2\big)\! \right)\! + R_{\textrm{A}}^2\! \left(\! \varphi_2 - \frac{1}{2}\sin\big(2\varphi_2\big)\! \right)\! ,\!
\end{align}
i.e., $\gamma$ is given in (\ref{eqL1-1}) with $\theta_2$ and $\varphi_2$ given in (\ref{eqL1-2}) and (\ref{eqL1-3}).
	
\emph{2)}~If $0<r<R_{\textrm{U}}-m_0$ and $R_{\textrm{U}}-R_{\textrm{A}}<m_0<R_{\textrm{U}}$, the PDF of $r$ can be expressed as
\begin{align}\label{eqApA5} 
  f_{R^{2)}}(r|m_0) =& \frac{2\pi r}{\gamma}.
\end{align} 
  
\emph{3)}~However, if $R_{\textrm{U}}\! -\! m_0\! <\! r\! <\! R_{\textrm{A}}$ and $R_{\textrm{U}}\! -\! R_{\textrm{A}}\! <\! m_0\! <\! R_{\textrm{U}}$, the overlapping region contains only a segment of the circle defined by the equation $x^2+y^2=r^2$. The abscissa $x_1^{*}$, which represents the intersection point between the circle $x^2+y^2=r^2$ and the circle $(x-m_0)^2+y^2=R_{\textrm{U}}^2$, is given by 
\begin{align}\label{eqApA6} 
	x_1^{*}=\frac{m_0^2-R_{\textrm{U}}^2+r^2}{2m_0}.
\end{align}
Thus the CDF of $r$ can be written as 
\begin{align}\label{eqApA7} 
	& F_{R^{3)}}(r|m_{0}) = \frac{1}{\gamma}\Bigg( \int_{m_{0}-R_{\textrm{U}}}^{x_{1}^{*}}2\sqrt{R_{U}^{2}-(x-m_0)^{2}}dx  \nonumber \\
	& \hspace*{20mm} + \int_{{x_{1}^{*}}}^{r}2\sqrt{r^{2}-x^{2}}dx\Bigg) \nonumber \\
	&\hspace*{4mm} =\frac{1}{\gamma}\! \left(\!\! R_{\textrm{U}}^2\! \bigg(\! \theta_4\! -\! \frac{1}{2}\sin\big(2\theta_4\big)\!\! \bigg)\! +\! r^2\!\bigg(\! \varphi_3\! -\! \frac{1}{2}\sin\big(2\varphi_3\big)\! \bigg)\!\! \right)\! ,\!
\end{align}
where $\theta_4=\arccos\left(\frac{m_0+R_{\textrm{U}}^2-r^2}{2m_0R_{\textrm{U}}}\right)$ and $\varphi_3$ is given in (\ref{eqL1-4}). The corresponding PDF $f_{R^{3)}}(r|m_0)$ is then given by (\ref{Lemma1.1}).

\emph{4)}~Similarly, if $m_0-R_{\textrm{U}}<r<R_{\textrm{A}}$ and $R_{\textrm{U}}<m_0<R_{\textrm{U}}+R_{\textrm{A}}$, as shown in Fig.~\ref{fr4}, we can easily obtain the CDF of r as
\begin{align}\label{eqApA8} 
	& F_{R^{4)}}(r|m_{0}) = \frac{1}{\gamma}\Bigg( \int_{m_{0}-R_{\textrm{U}}}^{x_{1}^{*}}2\sqrt{R_{U}^{2}-(x-m_0)^{2}}dx  \nonumber \\
	& \hspace*{20mm} + \int_{{x_{1}^{*}}}^{r}2\sqrt{r^{2}-x^{2}}dx\Bigg) \nonumber
\end{align}	
\begin{align} 
	&\hspace*{4mm} =\frac{1}{\gamma}\! \left(\!\! R_{\textrm{U}}^2\! \bigg(\! \theta_4\! -\! \frac{1}{2}\sin\big(2\theta_4\big)\!\! \bigg)\! +\! r^2\!\bigg(\! \varphi_3\! -\! \frac{1}{2}\sin\big(2\varphi_3\big)\! \bigg)\!\! \right)\! ,\!
\end{align}
where $\theta_4=\arccos\left(\frac{m_0+R_{\textrm{U}}^2-r^2}{2m_0R_{\textrm{U}}}\right)$ and $\varphi_3$ is given in (\ref{eqL1-4}). Thus, the corresponding PDF $f_{R^{4)}}(r|m_0)$ is given in (\ref{Lemma1.1}). 
This completes the proof.
\end{proof}

\subsection{Proof of Theorem~\ref{The1}}\label{ApB}

\begin{proof}
\begin{align}\label{eqApB1} 
  P_{\textrm{cov}}^{\textrm{T-A}} =& \mathbb{P}\left(\mathrm{SINR}_1\geq T_{h_1}\right) \nonumber \\
  =& \mathbb{E}_{R_m}\left[ \mathbb{P}\left(\mathrm{SINR}_1 \geq T_{h_1}|R_m=r_m,m_0\right) \right] \nonumber \\
  =& \int \int \mathbb{P}\left(\mathrm{SINR}_1 \geq T_{h_1}|R_m=r_m,m_0\right) \nonumber \\
  &\times f_{R_m}(r_m|m_0)\,f_M(m_0) \,\mathrm{d}r_m \,\mathrm{d}m_0.
\end{align}
We note that $\mathbb{P}\left(\mathrm{SINR}_1 \geq T_{h_1}|R_m=r_m,m_0\right)$ satisfies:
\begin{align}\label{eqApB2} 
	& \text{$\mathbb{P}\left(\mathrm{SINR}_1 \geq T_{h_1}|R_m=r_m,m_0\right)$} \nonumber \\
	& = 1 - \mathbb{P}\left(|h_{T_m\!A}|^2 \le \frac{T_{h_1}{I_{\textrm{T}}}(H_{\textrm{A}}^2+r_m^2)^{\frac{\alpha_1}{2}}}{	P_{\textrm{T}}} \right)\nonumber \\
	& \overset{\mathrm{(a)}}{<}  1 - \mathbb{E}\left[\left(1 - \exp\left(\frac{-\eta T_{h_1}{I_{\textrm{T}}}(H_{\textrm{A}}^2+r_m^2)^{\frac{\alpha_1}{2}}}{	P_{\textrm{T}}} \right) \right)^{N_{\textrm{TA}}} \right] \nonumber \\
	& \overset{\mathrm{(b)}}{=} \sum_{n=1}^{N_{\textrm{TA}}}(-1)^{n+1}\binom{N_{\textrm{TA}}}{n}\mathbb{E}_I\!\!\left[\exp\!\left(\! \frac{-n\eta T_{h_1}{I_{\textrm{T}}}(H_{\textrm{A}}^2+r_m^2)^{\frac{\alpha_1}{2}}}{	P_{\textrm{T}}}\right) \!\right] \nonumber \\
	& \overset{\mathrm{(c)}}{=} \sum_{n=1}^{N_{\textrm{TA}}}(-1)^{n+1}\binom{N_{\textrm{TA}}}{n}\mathbb{E}_I\left[ \exp\left(-s{I_{\textrm{T}}} \right)\right],
 \end{align}
where (a) is a tight upper bound when $N_{\textrm{TA}}$ is small \cite{alzer1997some}, that is, for small $N_{\textrm{TA}}$, $\mathbb{P}\left(|h|^{2}<\psi\right)<\mathbb{E}\left[\left(1 - \exp (-\psi\eta)\right)^{N_{\textrm{TA}}}\right]$ with $\eta=N_{\textrm{TA}}(N_{\textrm{TA}}!)^{-\frac{1}{N_{\textrm{TA}}}}$, and (b) is obtained by the binomial theorem, while (c) is obtained by denoting $s= \frac{n\eta T_{h_1}(H_{\textrm{A}}^2+r_m^2)^{\frac{\alpha_1}{2}}} {	P_{\textrm{T}}} $.  This completes the proof.
\end{proof}

\subsection{Proof of Lemma~\ref{L3}}\label{ApC}

\begin{proof}
Assumed that the distance denoted as $R_n$ between the interference user and the projection point is equal to $r_n$. Thus, $\mathcal{L}_{I_{\textrm{T}}}\left(s\right)$ can be written as
\begin{align} \label{eqApC1} 
	& \mathcal{L}_{I_{\textrm{T}}}(s) = \mathbb{E}\! \left[\! \prod_{T_n\in\Phi_{\textrm{U}} {\backslash} T_m}\!\!\! \exp\!\left( -s 	P_{\textrm{T}}\left|h_{T_n\!A}\right|^2 (H_{\textrm{A}}^2+r_n^2)^{-\frac{\alpha_1}{2}} \right)\! \right] \nonumber \\
		& = \mathbb{E}_{N_I,R_n}\!\Bigg[\! \prod_{T_n\in\Phi_{\textrm{U}} {\backslash} T_m}\!\!\! \mathbb{E}_{|h_{T_n\!A}|^2}\bigg[\exp\bigg(\left|h_{T_n\!A}\right|^2 \nonumber \\
	& \hspace*{15mm} \times \left(-s 	P_{\textrm{T}} (H_{\textrm{A}}^2+r_n^2)^{-\frac{\alpha_1}{2}}\right) \bigg)\! \bigg]\Bigg] \nonumber \\
	& \overset{\mathrm{(a)}}{=} \mathbb{E}_{N_I,R_n}\!\!\!\left[ \!\prod_{T_n\!\in\Phi_{\textrm{U}}  \! {\backslash}\! T_m}    \underset{D_2}{\underbrace{\left(\!1\!+\!\frac{s	P_{\textrm{T}}}{N_{\textrm{TA}}(H_{\textrm{A}}^2+r_n^2)^{\frac{\alpha_1}{2}}}\!\right)^{\!\!\!-N_{\textrm{TA}}}}} \!\right]\!\!,
\end{align}
where (a) is obtained by using the moment-generating function (MGF) of the normalized Gamma random	variable. This completes the proof.
\end{proof}

\subsection{Proof of Theorem~\ref{T2}}\label{ApD}

\begin{proof}
Using the Kummer’s transform of the hypergeometric function \cite{5313912}, the PDF of \emph{$|h|^{2}$} (\ref{eqPDFscf}) can be rewritten as
\begin{align}\label{eqApD1} 
  f_{|h|^{2}}(x) =& \sum\limits_{k=0}^{\infty} \Psi (k) x^{k} \exp\left(-(\beta-\delta) x\right) ,
\end{align}
where 
\begin{align}\label{eqApD2} 
  \Psi (k) =& \frac{(-1)^{k} \kappa \delta^{k}}{(k!)^{2}}\, (1-q)_{k} .
\end{align}
Then, the CDF of \emph{$|h|^{2}$} can be expressed as
\begin{align}\label{eqApD3} 
  F_{|h|^{2}}(x) =& \sum_{k=0}^{\infty} \Psi (k) \int_{0}^{x} t^{k} \exp\left(-(\beta-\delta)t\right) {\rm d} t \nonumber \\
	=& \sum_{k=0}^{\infty} \frac{\Psi (k)}{\left(\beta-\delta\right)^{k+1}} \bar{\gamma}\left(k+1,\left(\beta-\delta\right)x\right).
\end{align}
Therefore, we obtain $P_{\textrm{cov}}^{\textrm{A-S}}$ as
\begin{align}\label{eqApD4} 
	& P_{\textrm{cov}}^{\textrm{A-S}} \triangleq \mathbb{P}\left( \frac{P_m G_{A_m\!-\!S} \left|h_{A_m\!S}\right|^2 d_{m}^{-\alpha_2}}{I_{\textrm{A}}} \geq T_{h_2} \right) \nonumber \\
	& = 1 - \mathbb{P}\left( \left|h_{A_m\!S}\right|^2 \le \frac{T_{h_2} I_{\textrm{A}} d_{m}^{\alpha_2}}{P_m G_{A_m\!-\!S}}  \right) \nonumber \\
	& = 1 - \mathbb{E}\left[ \kappa \sum_{k=0}^{\infty} \frac{\Psi (k)}{(\beta-\delta)^{k+1}} \bar{\gamma}\left(k+1,(\beta\!-\!\delta) \frac{T_{h_2} I_{\textrm{A}} d_{m}^{\alpha_2}}{P_m G_{A_m\!-\!S}} \right) \right] \nonumber \\
		& \overset{\mathrm{(a)}}{\approx} 1 - \mathbb{E}\Bigg[ \sum_{k=0}^{\infty} \frac{\Psi (k)}{(\beta-\delta)^{k+1}}\Gamma (k+1) \nonumber \\
			&	\hspace*{5mm}\times \left( 1 - \exp\left(- \frac{\zeta(\beta-\delta) T_{h_2} I_{\textrm{A}} d_{m}^{\alpha_2}}{P_m G_{A_m\!-\!S}} \right)\right)^{k+1}\Bigg]  \nonumber \\
				\hspace*{-10pt}& \overset{\mathrm{(b)}}{=} 1 - \sum_{k=0}^{\infty} \frac{\Psi (k)}{(\beta-\delta)^{k+1}} \Gamma (k+1) \sum_{t=0}^{k+1} \binom{k+1}{t} \nonumber \\
			&	\hspace*{5mm}\times (-1)^t \mathbb{E}\left[\exp\big(-{s}'I_{\textrm{A}}\big)\right] , 
\end{align}
where (a) is approximated by using $\bar{\gamma}(k+1,x) < \Gamma (k+1)(1-\exp(-\zeta x))^{k+1}$ \cite{alzer1997some}, $\zeta=(\Gamma(k+2))^{-\frac{1}{k+1}}$,  and (b) is obtained from  binomial theorem with ${s}'=\frac{t\,\zeta(\beta-\delta)\, T_{h_2} \,d_{m}^{\alpha_2}}{P_m G_{A_m\!-\!S}}$. This completes the proof.
\end{proof}

\subsection{Proof of Lemma~\ref{L4}}\label{ApE}

\begin{proof}
\begin{align}\label{eqApE1} 
  & \mathcal{L}_{I_{\textrm{A}}}({s}')\! =\! \mathbb{E}\! \left[\exp\! \left(\! -{s}'\!\! \sum_{A_n\in\Phi_{\textrm{A}} {\backslash} A_m}\!\! P_{n} G_{A_n\!-\!S} \left|h_{A_n\!S}\right|^2 d_{n}^{-\alpha_2}\! \right)\! \right] \nonumber \\
 	& =\! \mathbb{E}\! \left[\! \prod_{A_n\in\Phi_{\textrm{A}} {\backslash} A_m}\!\!\!\!\! \mathbb{E}_{|h_{A_n\!S}|^2}\! \left[\exp\!\left(-{s}' P_{n} G_{A_n\!-\!S} \left|h_{A_n\!S}\right|^2 d_{n}^{-\alpha_2}\right)\! \right]\! \right] \nonumber \\
 	& \overset{\mathrm{(a)}}=\! \mathbb{E}_{N_{\textrm{A}}}\!\!\left[\! \prod_{A_n\in\Phi_{\textrm{A}} \! {\backslash}\! A_m}\!\!\!\!\!\! \mathbb{E}_{G_{A_n\!-\!S},|h_{A_n\!S}|^2} \!\! \left[\exp\left(-t_A \left|h_{A_n\!S}\right|^2 \right)\! \right] \! \right]\! ,\!
\end{align}
where (a) is obtained becasue all aerial transmitters have an equal distance to $S$, and $t_A={s}' P_{n} G_{A_n\!-\!S} d_{0}^{-\alpha_2}$.

As shown in \cite{2003A}, the MGF of the SR fading model is defined as $M_S(x)=\mathbb{E}\left[\exp(-x S)\right]=	\frac{(2 c q)^q(1+2 c x)^{q-1}}{((2 c q+\Omega )(1+2 c x)-\Omega)^q}$. Thus, we further obtain
\begin{align}
  & \mathcal{L}_{I_{\textrm{A}}}\!\left({s}'\right)\! =\! \mathbb{E}_{N_{\textrm{A}}}\!\! \left[\! \prod_{A_n\in\Phi_{\textrm{A}} {\backslash}\! A_m}\!\!\!\!\!\!\! \mathbb{E}_{G_{A_n\!-\!S}}\!\! \left[\! \underset{M_1(t_A)}{\underbrace{\frac{(2 c q)^q(1 + 2 c t_A)^{q-1}}{((2 c q\! +\! \Omega )(1\! +\! 2 c t_A)\! -\! \Omega)^q}}}\! \right]\! \right] \nonumber 
\end{align}
\begin{align}\label{eqApE2} 
	& \overset{\mathrm{(b)}}= \mathbb{E}_{N_{\textrm{A}}}\!\! \left[\! \prod_{A_n\in\Phi_{\textrm{A}} {\backslash}\! A_m} \!\! \left(\! \underset{M_2}{\underbrace{M_1(t_1)\frac{\theta}{2\pi}\! +\! M_1(t_2)\Big(1\!-\!\frac{\theta}{2\pi}\Big)}} \! \right) \! \right] \nonumber \\
	& \overset{\mathrm{(c)}}= \sum_{n=0}^{\infty} \frac{(\lambda_{\textrm{A}}S_{\textrm{A}})^{n}}{n!} \exp(-\lambda_A S_{\textrm{A}})(M_2)^{n} \nonumber \\
  & = \exp\left(\lambda_{\textrm{A}} S_{\textrm{A}} (M_2 -1)\right) ,
\end{align}
where (b) is obtained by denoting $t_1 ={s}' P_n d_0^{-\alpha_2} 	G_{\textrm{t}} 	G_{\textrm{r}}$ and $t_2 ={s}' P_n d_0^{-\alpha_2}	g_{\textrm{t}} 	G_{\textrm{r}}$, and (c) is obtained by the fact that $N_{\textrm{A}}$ follows the PPP with the density of $\lambda_{\textrm{A}}$ and $S_{\textrm{A}}=\pi R_{\textrm{C}}^2$. This completes the proof.
\end{proof}	
\subsection{Proof of Theorem~\ref{T3}}\label{ApF}
\begin{proof}
	From (\ref{eqAER}), we have
	\begin{align}\label{eqA6} 
	 \bar{C}^{\textrm{T-A}} 	&\triangleq \frac{1}{K}\mathbb{E}\left[\log_2(1 + \mathrm{SINR_1})\right] \nonumber \\
		& \overset{\mathrm{(a)}}=\! \frac{1}{K}\int \!\!\int \!\! \int_{t>0}\!\!\!\! \mathbb{E}\left[\mathbb{P}\!\!\ \Big(\log_2 (1+ \mathrm{SINR_1} >t\Big)|R_{m}= r_{m}, m_0 \right]\! \nonumber\\
		&\hspace*{4mm}\times f_{R_{m}}(r_m|m_0) f_M(m_0)\,\mathrm{d}t\,\mathrm{d}r_{m}\,\mathrm{d}m_{0},
	\end{align}	
	where (a)~follows from the fact that if the random variable $X$ involved is positive, $\mathbb{E}[X]\! =\! \int_{t>0}\mathbb{P}(X>t)\mathrm{d}t$.
	
	We note that $\mathbb{E}\left[\mathbb{P}\!\!\ \Big(\log_2 (1+ \mathrm{SINR_1} >t\Big)|R_{m}= r_{m}, m_0 \right]$ satisfies:
	\begin{align}
		&\mathbb{E}\left[\mathbb{P}\!\!\ \Big(\log_2 (1+ \mathrm{SINR_1} >t\Big)|R_{m}= r_{m}, m_0 \right] \nonumber\\
		&=   \mathbb{E}\! \left[\mathbb{P}\left( |h_{T_m\!A}|^2\! >\! \frac{I_{\textrm{T}}(2^t\! -\! 1)}{P_{\textrm{T}} (H_{\textrm{A}}^2+r_m^2)^{-\frac{\alpha_1}{2}}}\right)|R_{m}= r_{m}, m_0\right] \nonumber \\
		&\overset{\mathrm{(b)}}{<}  \left(1-\mathbb{E}\!\left[\!\left(\!1\!-\!\exp\!\left( \frac{-\eta {I_{\textrm{T}}}(2^t\! -\! 1) } {P_{\textrm{T}} (H_{\textrm{A}}^2+r_m^2)^{-\frac{\alpha_1}{2}}}  \right)\! \right)^{\!\! N_{\textrm{TA}}} \right] \right) \nonumber \\
		&\overset{\mathrm{(c)}}{=}  \sum_{n=1}^{N_{\textrm{TA}}}(-1)^{n+1}\binom{N_{\textrm{TA}}}{n}\mathbb{E}_{I_{\textrm{T}}}\! \! \left[\exp\left(-s_1 {I_{\textrm{T}}}\right)\right], 
	\end{align}
where (b)~is obtained by the tight upper bound when $N_{\textrm{TA}}$ is small \cite{alzer1997some}, that is, $\mathbb{P}\left[|h|^{2}\! <\! \psi\right]\! <\! \left(1\! -\! \exp(-\psi\eta)\right)^{N_{\textrm{TA}}}$ with $\eta\! =\! N_{\textrm{TA}}(N_{\textrm{TA}}!)^{-\frac{1}{N_{\textrm{TA}}}}$, and (c)~is obtained by the binomial theorem and by denoting $s_1\! =\! \frac{n \eta (2^t\! -\! 1) } {P_{\textrm{T}} (H_{\textrm{A}}^2+r_m^2)^{-\frac{\alpha_1}{2}}}$.  This completes the proof.
\end{proof}	

\subsection{Proof of Theorem~\ref{T4}}\label{ApG}
\begin{proof}
	
	Starting from the definition (\ref{eqAER}), we have
		\begin{align}  \label{eqA8}
		& \bar{C}^{\textrm{A-S}} \triangleq \frac{1}{K}\mathbb{E}\left[\log_2(1 + \mathrm{SINR}_2)\right] \nonumber \\
		& \overset{\mathrm{(a)}}=\! \frac{1}{K}\mathbb{E}\!\!\left[\int_{t>0}\!\! \mathbb{P}\Big(\!\log_2 (1\!+\!\frac{P_{\textrm{m}} G_{A_m\!-\!S}\!\left|h_{A_m\!S}\right|^2 \!d_m^{-\alpha_2}} {I_{\textrm{A}}} \!>\!t\!\Big)\mathrm{d}t \!\right] \nonumber \\
		& \overset{\mathrm{(b)}}= \frac{1}{K}\mathbb{E}\! \left[\int_{t>0}\!\!\!\left( 1- \mathbb{P}\Big( \left|h_{A_m\!S}\right|^2\! \leq \! \frac{ I_{\textrm{A}}(2^t\! -\! 1)}{P_{\textrm{m}} G_{A_m\!-\!S}  d_0^{-\alpha_2}}\Big) \right) \mathrm{d}t \right]\!  \nonumber \\
		& \overset{\mathrm{(c)}}= \frac{1}{K}\int_{t>0}\! \Bigg(1\! - \! \mathbb{E} \Bigg[ \sum_{k=0}^{\infty} \frac{\Psi\left(k\right)}{(\beta-\delta)^{\,k+1}} \nonumber\\
		&\hspace*{5mm}\times\!\gamma\! \left(\! k\! + \!1,(\beta\! -\! \delta) \frac{ I_{\textrm{A}}(2^t\! -\! 1)d_0^{\alpha_2}}{P_{\textrm{m}} G_{A_m\!-\!S}  }\! \right)\! \Bigg]    \Bigg) \mathrm{d}t \nonumber\\
		& \overset{\mathrm{(d)}}{\approx} \frac{1}{K}\int_{t>0}\! \Bigg(1 - \mathbb{E}\Bigg[ \sum_{k=0}^{\infty} \frac{\Psi (k)}{(\beta-\delta)^{k+1}}\Gamma (k+1) \nonumber \\
		&	\hspace*{5mm}\times \left( 1 - \exp\left(- \frac{\zeta(\beta-\delta) (2^t-1) I_{\textrm{A}} d_{0}^{\alpha_2}}{P_m G_{A_m\!-\!S}} \right)\right)^{k+1}\Bigg] \Bigg) \mathrm{d}t \nonumber \\ 
		&\overset{\mathrm{(e)}}= \frac{1}{K}\int_{t>0}\! \Bigg(1 - \sum_{k=0}^{\infty} \frac{\Psi (k)}{(\beta-\delta)^{k+1}} \Gamma (k+1) \sum_{v=0}^{k+1} \binom{k+1}{v} \nonumber \\
		&	\hspace*{5mm}\times (-1)^v \mathbb{E}\left[\exp\big(-{s_1}'I_{\textrm{A}}\big)\right] \Bigg)\!\mathrm{d}t,   
	\end{align}
	where (a)~follows from the fact that if the random variable $X$ involved is positive, $\mathbb{E}[X]\! =\! \int_{t>0}\mathbb{P}(X\! >\! t)\mathrm{d}t$, (b) is due to the fact that (d)~is approximately a constant $d_0$, (c) is obtained by substituting (\ref{eqApD3}) into right-hand side expression of (b), (d) is approximated by using $\gamma(k+1,x)\! <\! \Gamma (k+1)(1-\exp(-\zeta x))^{k+1}$ with $\zeta\! =\! (\Gamma(k\! +\! 2))^{-\frac{1}{k+1}}$ \cite{alzer1997some}, and (e)~is obtained from the binomial theorem with ${s}_1'\! =\! \frac{v\,\zeta(\beta-\delta)\, (2^t-1) \,d_{0}^{\alpha_2}}{P_m G_{A_m\!-\!S}}$. 
\end{proof}	
\small

\begin{thebibliography}{10}
\providecommand{\url}[1]{#1}
\csname url@samestyle\endcsname
\providecommand{\newblock}{\relax}
\providecommand{\bibinfo}[2]{#2}
\providecommand{\BIBentrySTDinterwordspacing}{\spaceskip=0pt\relax}
\providecommand{\BIBentryALTinterwordstretchfactor}{4}
\providecommand{\BIBentryALTinterwordspacing}{\spaceskip=\fontdimen2\font plus
\BIBentryALTinterwordstretchfactor\fontdimen3\font minus
  \fontdimen4\font\relax}
\providecommand{\BIBforeignlanguage}[2]{{%
\expandafter\ifx\csname l@#1\endcsname\relax
\typeout{** WARNING: IEEEtran.bst: No hyphenation pattern has been}%
\typeout{** loaded for the language `#1'. Using the pattern for}%
\typeout{** the default language instead.}%
\else
\language=\csname l@#1\endcsname
\fi
#2}}
\providecommand{\BIBdecl}{\relax}
\BIBdecl

\bibitem{9042251} 
E.~Yaacoub and M.-S.~Alouini, ``A key 6G challenge and opportunity--connecting the base of the pyramid: A survey on rural connectivity,'' \emph{Proc. IEEE}, vol.~108, no.~4, pp.~533--582, Apr.~2020.

\bibitem{9314201} 
Q.~Huang, \emph{et al.}, ``Uplink massive access in mixed RF/FSO satellite-aerial-terrestrial networks,'' \emph{IEEE Trans. Commun.}, vol.~69, no.~4, pp.~2413--2426, Apr.~2021.

\bibitem{8353853} 
X.~Zhu, \emph{et al.}, ``Cooperative multigroup multicast transmission in integrated terrestrial-satellite networks,'' \emph{IEEE J. Sel. Areas Commun.}, vol.~36, no.~5, pp.~981--992, May 2018.

\bibitem{9693912} 
Y.~Zhang, \emph{et al.}, ``Resource allocation in terrestrial-satellite-based next generation multiple access networks with interference cooperation,'' \emph{IEEE J. Sel. Areas Commun.}, vol.~40, no.~4, pp.~1210--1221, Apr.~2022.

\bibitem{1522108} 
B.~Evans, \emph{et al.}, ``Integration of satellite and terrestrial systems in future multimedia communications,'' \emph{IEEE Wireless Commun}, vol.~12, no.~5, pp.~72--80, Oct.~2005.

\bibitem{9179999} 
G.~Pan, J.~Ye, Y.~Tian, and M.-S.~Alouini, ``On HARQ schemes in satellite-terrestrial transmissions,'' \emph{IEEE Trans. Wireless Commun.}, vol.~19, no.~12, pp.~7998--8010, Dec.~2020.

\bibitem{9610113} 
X.~Zhu and C.~Jiang, ``Integrated satellite-terrestrial networks toward 6G: Architectures, applications, and challenges,'' \emph{IEEE Internet Things J.}, vol.~9, no.~1, pp.~437--461, Jan.~2022.

\bibitem{9003405} 
Y.~Ruan, \emph{et al.}, ``Cooperative resource management for cognitive satellite-aerial-terrestrial integrated networks towards IoT,'' \emph{IEEE Access}, vol.~8, pp.~35759--35769, Feb.~2020.







\bibitem{huaicong2022ergodic} 
H.~Kong, \emph{et al.}, ``Ergodic sum rate for uplink NOMA transmission in satellite-aerial-ground integrated networks,'' \emph{Chinese J. Aeronaut.}, vol.~35, no.~9, pp.~58--70, Sep.~2022.

\bibitem{9808306} 
H.~Kong, \emph{et al.},  ``Hybrid multiple access transmission in satellite-aerial-terrestrial networks,'' \emph{IEEE Commun. Lett.}, vol.~26, no.~9, pp.~2146--2150, Sep.~2022.

\bibitem{9347980} 
A.~Yastrebova, \emph{et~al.}, ``Theoretical and simulation-based analysis of terrestrial interference to LEO satellite uplinks,'' in \emph{Proc. GLOBECOM 2020} (Taipei, Taiwan, China), Dec.~7-11, 2020, pp.~1--6.



\bibitem{9509510} 
B.~A.~Homssi and A.~Al-Hourani, ``Modeling uplink coverage performance in hybrid satellite-terrestrial networks,'' \emph{IEEE Commun. Lett.}, vol.~25, no.~10, pp.~3239--3243, Oct.~2021.
\bibitem{9676997} 
B. Manzoor, A. Al-Hourani, and B. A. Homssi, ``Improving iot-over-satellite connectivity using frame repetition technique,'' \emph{IEEE Wireless Commun. Lett.}, vol. 11, no. 4, pp. 736–740, Apr. 2022.
\bibitem{9838778} 
C.~C.~Chan, B.~Al~Homssi, and A.~Al-Hourani, ``A stochastic geometry approach for analyzing uplink performance for IoT-over-satellite,'' in \emph{Proc. ICC 2022} (Seoul, South Korea), May~16-20, 2022, pp.~2363--2368.

\bibitem{9130899} 
G.~Pan, J.~Ye, Y.~Zhang, and M.-S.~Alouini, ``Performance analysis and optimization of cooperative satellite-aerial-terrestrial systems,'' \emph{IEEE Trans. Wireless Commun.}, vol.~19, no.~10, pp.~6693--6707, Oct.~2020.

\bibitem{song2022cooperative} 
Z.~Song, \emph{et~al.}, ``Cooperative satellite-aerial-terrestrial systems: A stochastic geometry model,'' \emph{IEEE Trans. Wireless Commun.}, vol.~22, no.~1, pp.~220--236, Jan.~2023.

\bibitem{9789274} 
H.~Zhang, \emph{et~al.}, ``Outage analysis of cooperative satellite-aerial-terrestrial networks with spatially random terminals,'' \emph{IEEE Trans. Commun.}, vol.~70, no.~7, pp.~4972--4987, Jul.~2022.

\bibitem{9582189} 
D.~Vasudha, J.~Asritha, A.~Maloo, and P.~K.~Sharma, ``Coverage analysis of cooperative satellite-UAV-terrestrial communication system with receiver mobility,'' in \emph{Proc. INDISCON 2021} (Nagpur, India), Aug.~27-29, 2021, pp.~1--5.

\bibitem{7882710} 
M.~Afshang and H.~S.~Dhillon, ``Fundamentals of modeling finite wireless networks using binomial point process,'' \emph{IEEE Trans. Wireless Commun.}, vol.~16, no.~5, pp.~3355--3370, May 2017.

\bibitem{6574907} 
B.~Cho, K.~Koufos, and R.~Jantti, ``Bounding the mean interference in Mat\'ern type II hard-core wireless networks,'' \emph{IEEE Wireless Commun. Lett.}, vol.~2, no.~5, pp.~563--566, Oct.~2013.


\bibitem{6932503} 
T.~Bai and R.~W.~Heath, ``Coverage and rate analysis for millimeter-wave cellular networks,'' \emph{IEEE Trans. Wireless Commun.}, vol.~14, no.~2, pp.~1100--1114, Feb.~2015.

\bibitem{8016632} 
W.~Yi, Y.~Liu, and A.~Nallanathan, ``Modeling and analysis of D2D millimeter-wave networks with Poisson cluster processes,'' \emph{IEEE Trans. Commun.}, vol.~65, no.~12, pp.~5574--5588, Dec.~2017.

\bibitem{Zhao_etal2019} 
G. Zhao, \emph{et al.}, ``Mobile-traffic-aware offloading for energy- and spectral-efficient large-scale D2D-enabled cellular networks,'' \emph{IEEE Trans. Wireless Commun.}, vol.~18, no.~6, pp.~3251--3264, Jun.~2019.

\bibitem{1623307} 
C.~Loo, ``A statistical model for a land mobile satellite link,'' \emph{IEEE Trans. Veh. Technol.}, vol.~34, no.~3, pp.~122--127, Aug. 1985.

\bibitem{8068989} 
O.~Y.~Kolawole, S.~Vuppala, M.~Sellathurai, and T.~Ratnarajah, ``On the performance of cognitive satellite-terrestrial networks,'' \emph{IEEE Trans. Cogn. Commun. Netw}, vol.~3, no.~4, pp.~668--683, Dec.~2017.

\bibitem{2021Stochastic} 
A.~Talgat, M.~A.~Kishk, and M.~S.~Alouini, ``Stochastic geometry-based analysis of LEO satellite communication systems,'' \emph{IEEE Commun. Lett.}, vol.~25, no.~8, pp.~2458--2462, Aug.~2021.

\bibitem{9678973} 
D.-H.~Jung, J.-G.~Ryu, W.-J.~Byun, and J.~Choi, ``Performance analysis of satellite communication system under the shadowed-Rician fading: A stochastic geometry approach,'' \emph{IEEE Trans. Commun.}, vol.~70, no.~4, pp.~2707--2721, Apr.~2022.

\bibitem{bhatnagar2013closed} 
M.~R.~Bhatnagar and M.~Arti, ``On the closed-form performance analysis of maximal ratio combining in shadowed-Rician fading LMS channels,'' \emph{IEEE Commun. Lett.}, vol.~18, no.~1, pp.~54--57, Jan.~2014.

\bibitem{7869087} 
M.~Sellathurai, S.~Vuppala, and T.~Ratnarajah, ``User selection for multi-beam satellite channels: A stochastic geometry perspective,'' in \emph{Proc. 50th Asilomar Conf. Signals, Syst. Comput.} (Pacific Grove, CA, USA), Nov.~6-9, 2016, pp.~487--491.

\bibitem{8894851} 
Q.~Huang, \emph{et al.}, ``Performance analysis of integrated satellite-terrestrial multiantenna relay networks with multiuser scheduling,'' \emph{IEEE Trans. Aerosp. Electron. Syst.}, vol.~56, no.~4, pp.~2718--2731, Aug.~2020.



\bibitem{9520123} 
X.~Zhang, \emph{et al.}, ``Stochastic geometry-based analysis of cache-enabled hybrid satellite-aerial-terrestrial networks with non-orthogonal multiple access,'' \emph{IEEE Trans. Wireless Commun.}, vol.~21, no.~2, pp.~1272--1287, Feb.~2022.

\bibitem{9926973} 
A. G. Kanatas, ``Spherical Random Arrays With Application to Aerial Collaborative Beamforming,'' \emph{IEEE Trans. on Antennas and Propag.}, vol.~71, no.~1, pp.~550-562, Jan.~2023.

\bibitem{8059815} 
T. Van Luyen and T. Vu Bang Giang, ``Interference Suppression of ULA Antennas by Phase-Only Control Using Bat Algorithm,'' \emph{IEEE Antennas Wireless Propag. Lett.}, vol.~16, pp.~3038-3042,~2017.


\bibitem{2003A} 
A.~Abdi, W.~C.~Lau, M.-S.~Alouini, and M.~Kaveh, ``A new simple model for land mobile satellite channels: First- and second-order statistics,'' \emph{IEEE Trans. Commun.}, vol.~2, no.~3, pp.~519--528, May 2003.

\bibitem{jung2018outage} 
D.-H.~Jung and D.-G.~Oh, ``Outage performance of shared-band on-board processing satellite communication system,'' in \emph{Proc. VTC2018-Fall} (Chicago, IL, USA), Aug. 27-30, 2018, pp.~1--5.

\bibitem{zhang2019performance} 
X.~Zhang, \emph{et al.}, ``Performance analysis of NOMA-based cooperative spectrum sharing in hybrid satellite-terrestrial networks,'' \emph{IEEE Access}, vol.~7, pp.~172321--172329, Dec.~2019.

\bibitem{alzer1997some} 
H.~Alzer, ``On some inequalities for the incomplete gamma function,'' \emph{Math. Comput.}, vol.~66, no.~218, pp.~771--778, 1997.

\bibitem{5313912} 
A.~L.~Fructos, R.~R.~Boix, and F.~Mesa, ``Application of Kummer's transformation to the efficient computation of the 3-D Green's function with 1-D periodicity,'' \emph{IEEE Trans. Antennas Propag.}, vol.~58, no.~1, pp.~95--106, Jan.~2010.

\end{thebibliography}

\begin{IEEEbiography}
		[{\includegraphics[width=1in,height=1.25in,clip,keepaspectratio]{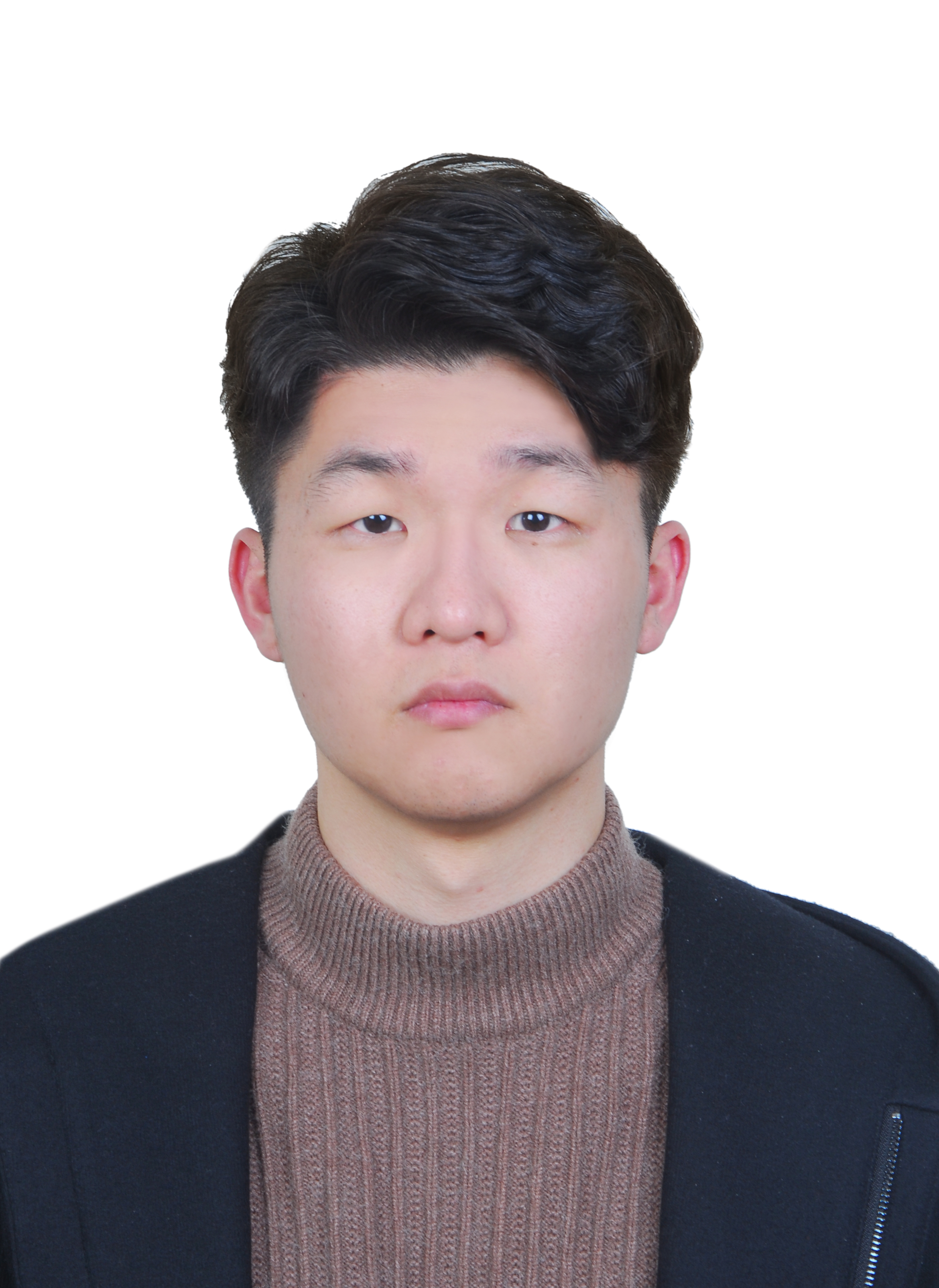}}]
		{Wen-Yu Dong}
	received the B.S. degree in electronic and information engineering from Sichuan  University (SCU), China, in 2019. He is currently pursuing the Ph.D. degree in information and communication engineering with the School of Information and Communication Engineering, Beijing University of Posts and Telecommunications (BUPT), and with the Key Laboratory of Universal Wireless Communications, Ministry of Education. His current research interest includes high-dynamic mobile ad hoc network architecture and protocol stack design, routing design for mobile ad hoc networks, and integrated space-air-ground network modeling and performance analysis based on stochastic geometry.
\end{IEEEbiography}
	
\begin{IEEEbiography}
	[{\includegraphics[width=1in,height=1.25in,clip,keepaspectratio]{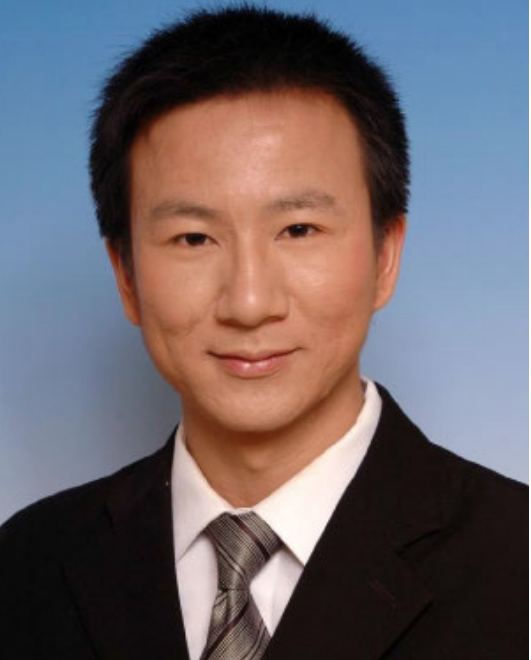}}]
	{Shaoshi Yang}
	   (Senior Member, IEEE) received the B.Eng. degree in information engineering from Beijing University of Posts and Telecommunications (BUPT), China, in 2006, and the Ph.D. degree in electronics and electrical engineering from Univer sity of Southampton, UK, in 2013. From 2008 to 2009, he was a Researcher with Intel Labs China. From 2013 to 2016, he was a Research Fellow with the School of Electronics and Computer Science, University of Southampton. From 2016 to 2018, he was a Principal Engineer with Huawei Technologies Co., Ltd., where he made significant contributions to the products, solutions and standardization of 5G, wideband IoT, and cloud gaming/VR. He was a Guest Researcher with the Isaac Newton Institute for Mathematical Sciences, University of Cambridge. He is currently a Full Professor with BUPT. His research interests include 5G/5G-A/6G, massive MIMO, mobile ad hoc networks, distributed artificial intelligence, and cloud gaming/VR. He is a standing committee member of the CCF Technical Committee on Distributed Computing and Systems. He received Dean’s Award for Early Career Research Excellence from University of Southampton in 2015, Huawei President Award for Wireless Innovations in 2018, IEEE TCGCC Best Journal Paper Award in 2019, IEEE Communications Society Best Survey Paper Award in 2020, CAI Invention and Entrepreneurship Award in 2023, CIUR Industry University-Research Cooperation and Innovation Award in 2023, and the First Prize of Beijing Municipal Science and Technology Advancement Award in 2023. He was/is an Editor of \emph{IEEE Transactions on Communications}, \emph{IEEE Systems Journal}, \emph{IEEE Wireless Communications Letters}, and \emph{Signal Processing} (Elsevier). For more details of his research progress, please refer to https://shaoshiyang.weebly.com/.
\end{IEEEbiography}

\begin{IEEEbiography}
	[{\includegraphics[width=1in,height=1.25in,clip,keepaspectratio]{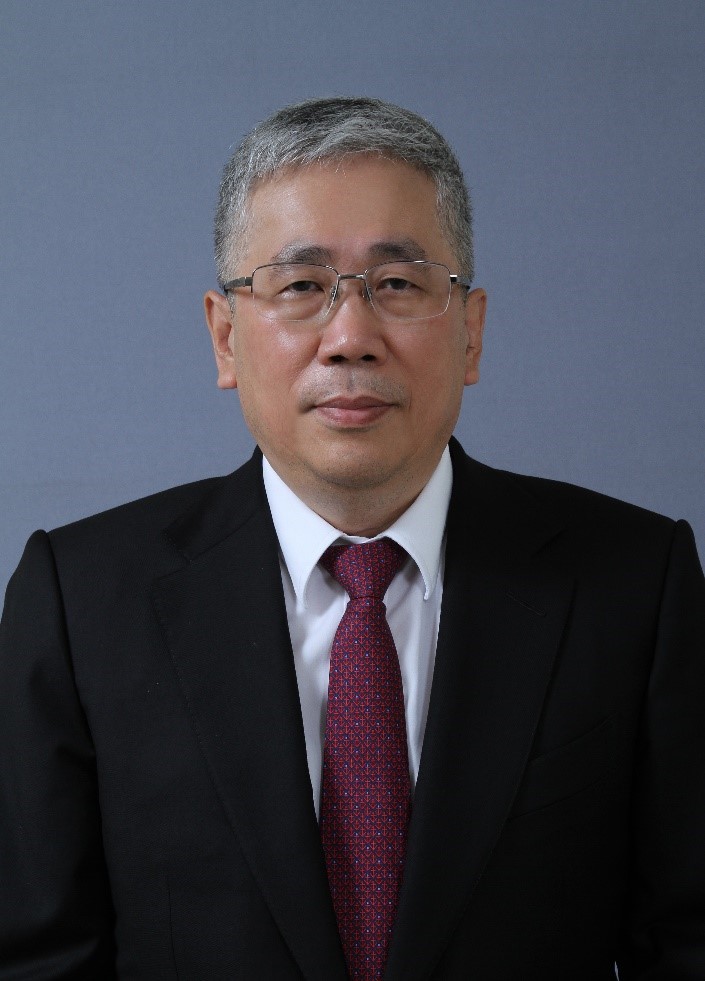}}]
	{Ping Zhang}
(Fellow, IEEE)  is currently a professor of School of Information and Communication Engineering at Beijing University of Posts and Telecommunications, the director of State Key Laboratory of Networking and Switching Technology, a member of IMT-2020 (5G) Experts Panel, a member of Experts Panel for China’s 6G development. He served as Chief Scientist of National Basic Research Program (973 Program), an expert in Information Technology Division of National High-tech R$\&$D program (863 Program), and a member of Consultant Committee on International Cooperation of National Natural Science Foundation of China. His research interests mainly focus on wireless communication. He is an Academician of the Chinese Academy of Engineering (CAE).
\end{IEEEbiography}

\begin{IEEEbiography}[{\includegraphics[width=1in,height=1.25in,clip,keepaspectratio]{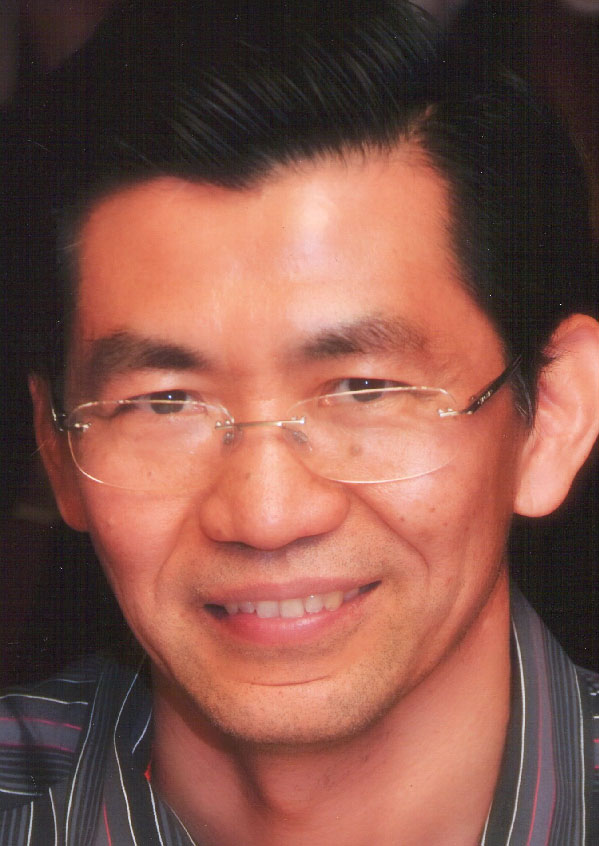}}]{Sheng Chen}
	
	(Life Fellow, IEEE) received his BEng degree from the East China Petroleum Institute, Dongying, China, in 1982, and his PhD degree from the City University, London, in 1986, both in control engineering. In 2005, he was awarded the higher doctoral degree, Doctor of Sciences (DSc), from the University of Southampton, Southampton, UK. From 1986 to 1999, He held research and academic appointments at the Universities of Sheffield, Edinburgh and Portsmouth, all in UK. Since 1999, he has been with the School of Electronics and Computer Science, the University of Southampton, UK, where he holds the post of Professor in Intelligent Systems and Signal Processing. Dr Chen's research interests include adaptive signal processing, wireless communications, modeling and identification of nonlinear systems, neural network and machine learning, evolutionary computation methods and optimization. He has published over 700 research papers. Professor Chen has 20,000+ Web of Science citations with h-index 61 and 39,000+ Google Scholar citations with h-index 83. Dr. Chen is a Fellow of the United Kingdom Royal Academy of Engineering, a Fellow of Asia-Pacific Artificial Intelligence Association and a Fellow of IET. He is one of the original ISI highly cited researchers in engineering (March 2004).
	
\end{IEEEbiography}

\end{document}